\documentclass[11pt,a4paper,onecolumn,accepted=2023-04-04]{quantumarticle}
\pdfoutput=1
\usepackage[utf8]{inputenc}
\usepackage{amsmath}
\usepackage{amsfonts}
\usepackage{amssymb}
\usepackage{amsthm}

\usepackage{booktabs}
\usepackage{subcaption}
\usepackage{hyperref}
\usepackage{enumitem}

\usepackage{tikz}
\usetikzlibrary{math}
\usepackage[compat=0.4]{yquant}

\usepackage[numbers,sort&compress]{natbib}

\usepackage[margin=1.1in]{geometry}
\title{Post-selection-free preparation of high-quality physical qubits}
\author{Ben Barber}
\email{ben.barber@riverlane.com}
\author{Neil I. Gillespie}
\email{neil.gillespie@riverlane.com}
\author{J. M. Taylor}
\email{jake.taylor@riverlane.com}
\thanks{current address Joint Center for Quantum Information and Computer Science}
\affiliation{Riverlane, Cambridge, UK}
\date{}

\newtheorem{theorem}{Theorem}
\newtheorem{lemma}[theorem]{Lemma}
\newtheorem{proposition}[theorem]{Proposition}

\newcommand{\ket}[1]{|#1\rangle}
\newcommand{\defn}{\emph}
\newcommand{\prep}{p_0}
\newcommand{\idle}{p_I}
\newcommand{\cnot}{p_C}
\newcommand{\toffoli}{p_T}
\newcommand{\pout}{p_{\text{out}}}
\newcommand{\threshold}{\theta}

\newcommand{\powerset}{\mathcal P}
\DeclareMathOperator{\Sym}{Sym}
\DeclareMathOperator{\Alt}{Alt}

\begin{document}

\maketitle

\begin{abstract}
Rapidly improving gate fidelities for coherent operations mean that errors in state preparation and measurement (SPAM) may become a dominant source of error for fault-tolerant operation of quantum computers. This is particularly acute in superconducting systems, where tradeoffs in measurement fidelity and qubit lifetimes have limited overall performance.
Fortunately, the essentially classical nature of preparation and measurement enables a wide variety of techniques for improving quality using auxiliary qubits combined with classical control and post-selection. In practice, however, post-selection greatly complicates the scheduling of processes such as syndrome extraction. 
Here we present a family of quantum circuits that prepare high-quality $\ket 0$ states without post-selection, instead using CNOT and Toffoli gates to non-linearly permute the computational basis.
We find meaningful performance enhancements when two-qubit gate fidelities errors go below 0.2\%, and even better performance when native Toffoli gates are available.
\end{abstract}

\noindent
Key physical implementations of quantum computers now reliably achieve two-qubit gate fidelities approaching 99.8\%~\cite{PhysRevX.11.021058}, with multiple systems reporting fidelities ranging from 99.2\% to 99.6\%~\cite{moskalenko2022high,PhysRevLett.125.150505,erhard2019characterizing,Google2019}.
This leaves SPAM (state preparation and measurement) errors as a dominant source of error in many, particularly superconducting, machines. 
Google's Sycamore processor experienced measurement error rates greater than $3$\%~\cite{Google2019}, and IBM report similar levels of measurement error~\cite{ibm_2020} in their devices.
Recent improvements to standard superconducting measurement error rates have not yet reduced them below 1\%~\cite{PhysRevX.11.011027}, although more is possible using multilevel qubit encodings~\cite{PhysRevX.10.011001}.
In contrast, for ions, neutral atoms, and spins in silicon measurement error rates can be much lower (cf. IonQ's 99.3\%~\cite{wright2019benchmarking} and HRL's 99.75\%~\cite{PRXQuantum.3.010352}), though achieving low rates in dense arrays remain a substantial challenge~\cite{PhysRevLett.127.050501}. 

There are three broad approaches to reducing the impact of SPAM errors.
The first is to improve the physical processes of preparation and measurement.
In most architectures there are tradeoffs between SPAM errors and speed, as longer integration or preparation times can improve performance but take longer; see for example~\cite{PhysRevX.10.011001} for the benefits of measuring slowly up to $T_1$ limits. 
The second is error mitigation techniques which calibrate and compensate for errors by running the circuit many times and reconstructing the expected outputs post-facto. 
This is helpful for NISQ applications~\cite{PhysRevLett.119.180509, Geller_2021} but does little to improve the entropy extraction critical for quantum error correction, and the extension of error mitigation into the early fault-tolerant regime~\cite{PRXQuantum.3.010345} requires SPAM errors with similar performance to two-qubit gate fidelities. 
Here we examine the third approach: algorithmically improving the quality of preparation and measurement using quantum logic gates.

Algorithmic improvement for measurement is straightforward.
The simplest example is used in, e.g., the double species ion clock (see~\cite{hume_clock,jake_clock}, and ~\cite{PhysRevLett.123.033201} for a recent example), where one readout qubit is used multiple times to estimate the clock qubit---a simple type of repetition code. 
More generally, CNOT gates can encode a set of measurement outcomes in a classical error correcting code on auxiliary measurement qubits, allowing recovery from sufficiently few measurement errors~\cite{PhysRevA.105.012419}. 
As these codes have very good performance and efficient decoders, we recommend their use but do not elaborate on them further here.

A conventional approach to improving preparation using post-selection is the circuit in Figure~\ref{fig:measurement-based}.
We prepare two qubits in state $\ket 0$, apply a CNOT from the first to the second, then measure the second in the computational basis.
If the measurement result is $0$, we output the first qubit as a high-quality $\ket 0$.
If not, we reject and begin the process again. 

\begin{figure}
\centering
  \subcaptionbox{post-selection\label{fig:measurement-based}}[0.4\linewidth]
{
\begin{tikzpicture}
  \begin{yquant}
    qubit {$q_{\idx}$} q[2];
    cnot q[1] | q[0];
    measure q[1];
  \end{yquant}
\end{tikzpicture}
}
\subcaptionbox{A $(3,1,1)$ post-selection-free purification circuit.\label{fig:(3,1,1)}}[0.4\linewidth]
{
\begin{tikzpicture}
  \begin{yquant}
    qubit {$q_{\idx}$} q[3];
    cnot q[1] | q[0];
    cnot q[2] | q[0];
    cnot q[0] | q[1], q[2];
  \end{yquant}
\end{tikzpicture}
}
    \caption{Two approaches to $\ket 0$ state production.}
    \label{fig:measurement-vs-logic}
\end{figure}
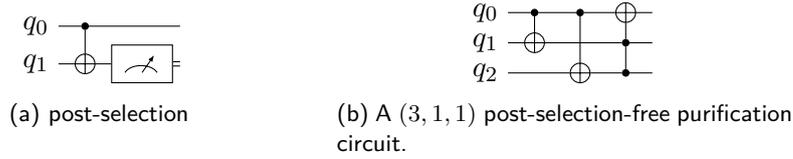

Suppose for now that gates are perfect, but that each preparation incorrectly produces $\ket 1$ with probability $\prep$, and that the measurement outcome is also incorrect with probability $\prep$, with all errors independent.
Then the measurement result is $0$ precisely when there are an even number of errors across the two preparations and one measurement.
We will correctly output $\ket 0$ if there are no errors, or two errors which occur on the second preparation and the measurement.
Otherwise we will either output $\ket 1$ or reject.
Conditional on acceptance, this circuit produces post-selected $\ket 0$ states with an error rate $2\prep^2 + O(\prep^3)$. 

Post-selection is an acceptable method for production of complex logical states such as $T$-magic states, where the scheduling difficulties it presents are the price we pay for implementing a logical non-Clifford gate fault-tolerantly.
However, it is much less satisfactory for the large number of simple physical $\ket 0$ states required for every round of syndrome extraction.
Specific challenges induced by post-selection include
\begin{itemize}
\item the need to schedule multiple attempts in case the first attempt fails; 
\item the need to wait for measurement results and classical processing, which in some systems can take appreciable time;
\item the risk of measurement inducing cross-talk errors on adjacent qubits;
\item limited advantage when measurement errors are substantially worse than preparation errors.
\end{itemize}

We propose an alternative approach avoiding post-selection, with the simplest example shown in Figure~\ref{fig:(3,1,1)}.
This circuit fixes $\ket{000}$, $\ket{010}$ and $\ket{001}$, and takes $\ket{100}$ to $\ket{011}$, so outputs $\ket 0$ on the first wire precisely when there is at most one error on the input.
Thus it produces $\ket 0$ states with slightly higher error rate $3\prep^2 + O(\prep^3)$, but does so without the use of post-selection.

\begin{table}
\centering
\begin{tabular}{ccc}\toprule
& \multicolumn{2}{c}{result} \\\cmidrule(lr){2-3}
total errors & post-selection & purification \\\midrule
$0$ & $\ket 0$ & $\ket 0$  \\
$1$ & reject & $\ket 0$  \\
$2$ & $\begin{cases}\ket 0 &  \text{with probability } 1/3 \\ \ket 1 & \text{with probability } 2/3\end{cases}$ & $\ket 1$  \\
$3$ & reject & $\ket 1$  \\\bottomrule
\end{tabular}
\caption{Comparison of two approaches to $\ket 0$ state production.}
\label{table:first-example}
\end{table}

The outcomes for each circuit are summarised in Table~\ref{table:first-example}.
Our post-selection-free approach has a higher yield of $\ket 0$ states, correctly processing all cases with $1$ error rather than one of the cases with $2$ errors.
It greatly simplifies scheduling of $\ket 0$ production, and avoids delays on architectures where measurement is significantly slower than applying gates.
On the other hand it has higher resource requirements, requiring one additional qubit and more gates, including a physical Toffoli.%
\footnote{A variety of constructions could instead be used to approximate a Toffoli gate up to relative phases.
For example, the circuit 
\begin{tikzpicture}[scale=0.5]
  \begin{yquant}
      qubit {$q_\idx$} q[3];
      h q[0] | q[1];
      zz (q[0,2]);
      h q[0] | q[1];
  \end{yquant}
\end{tikzpicture}
$q_1$-conditionally conjugates the $q_2$-conditional $Z$ gate to an $X$, resulting in an operation differing from the Toffoli with target $q_0$ by $-1$ precisely on the basis state $\ket{101}$.
In our model (see Section~\ref{sec:error-model}) the state is a mixture of computational basis states so these relative phases are unobservable.
}

We call the circuit in Figure~\ref{fig:(3,1,1)} a $(3,1,1)$ purification circuit because it uses three qubits to prepare one high-quality $\ket 0$ and can tolerate one error; we define the notation formally at the beginning of Section~\ref{sec:purification}.
For the $(3,1,1)$ circuit and a selection of other circuits described in more detail later, Figure~\ref{fig:intro-rates} shows the output error rate as a function of the input error rate, assuming that gates depolarise qubits with probability $0.003$ and idle qubits depolarise with probability $0.001$ in each round.
The $(3,1,1)$ circuit improves a 2\% preparation error rate to a 0.5\% error rate and a 1\% preparation error rate to a 0.4\% error rate.

Figure~\ref{fig:threshold-3} shows the thresholds a gate set must meet in order for the circuit of Figure~\ref{fig:(3,1,1)} to improve preparation quality at fixed idle depolarisation rate $0.001$.
Contours correspond to preparation error rates.
If your CNOT and Toffoli depolarisation rates place you to the left of a contour then you obtain an improvement from the $(3,1,1)$ circuit.

\begin{figure}
\centering
\includegraphics{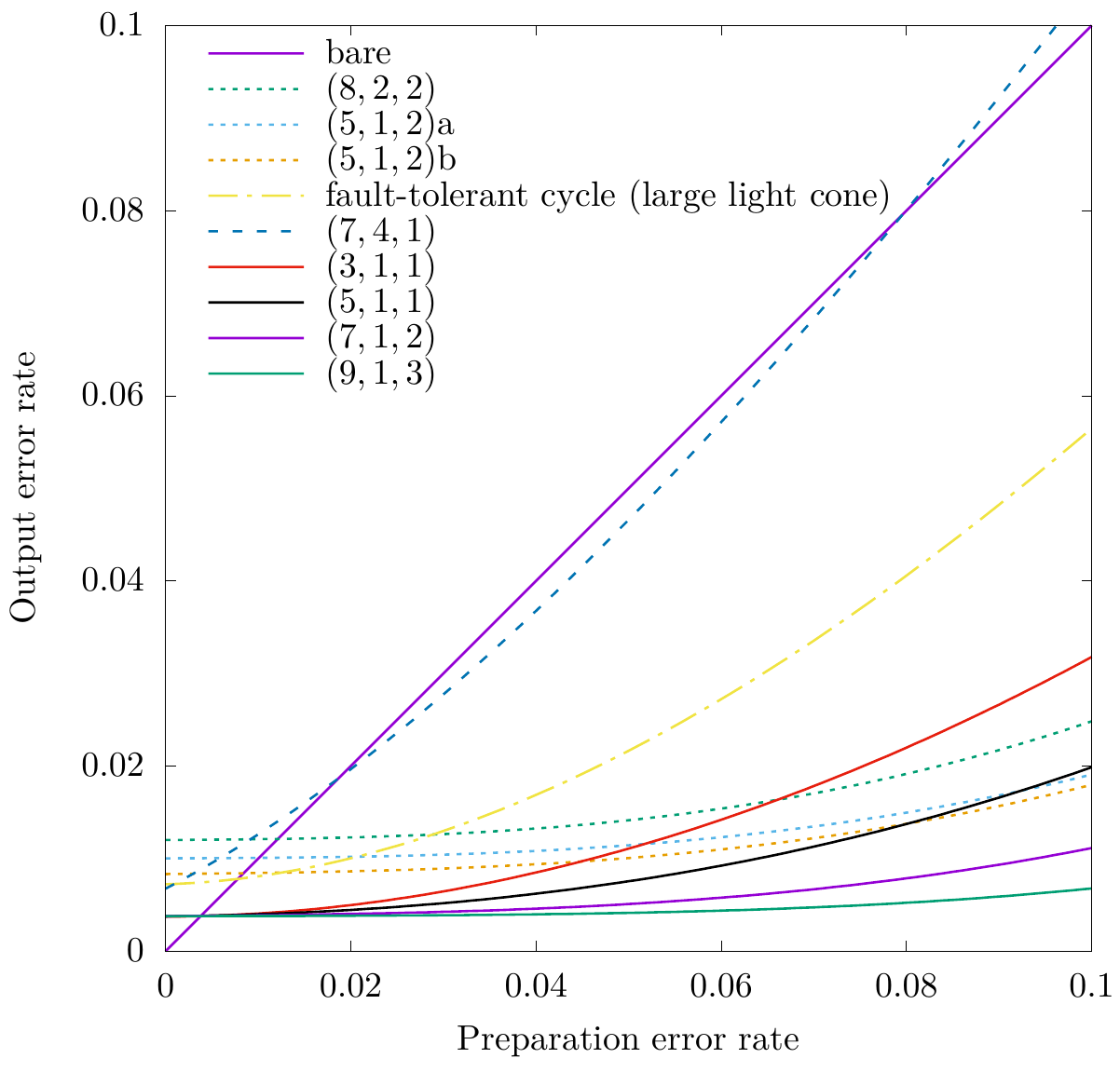}
\caption{Output error rate as a function of input error rate when gates depolarise with probability $0.003$ and idle qubits depolarise with probability $0.001$.  The curve labelled $(3,1,1)$ corresponds to the circuit of Figure~\ref{fig:(3,1,1)}.}
\label{fig:intro-rates}
\end{figure}

\begin{figure}
\centering
\includegraphics{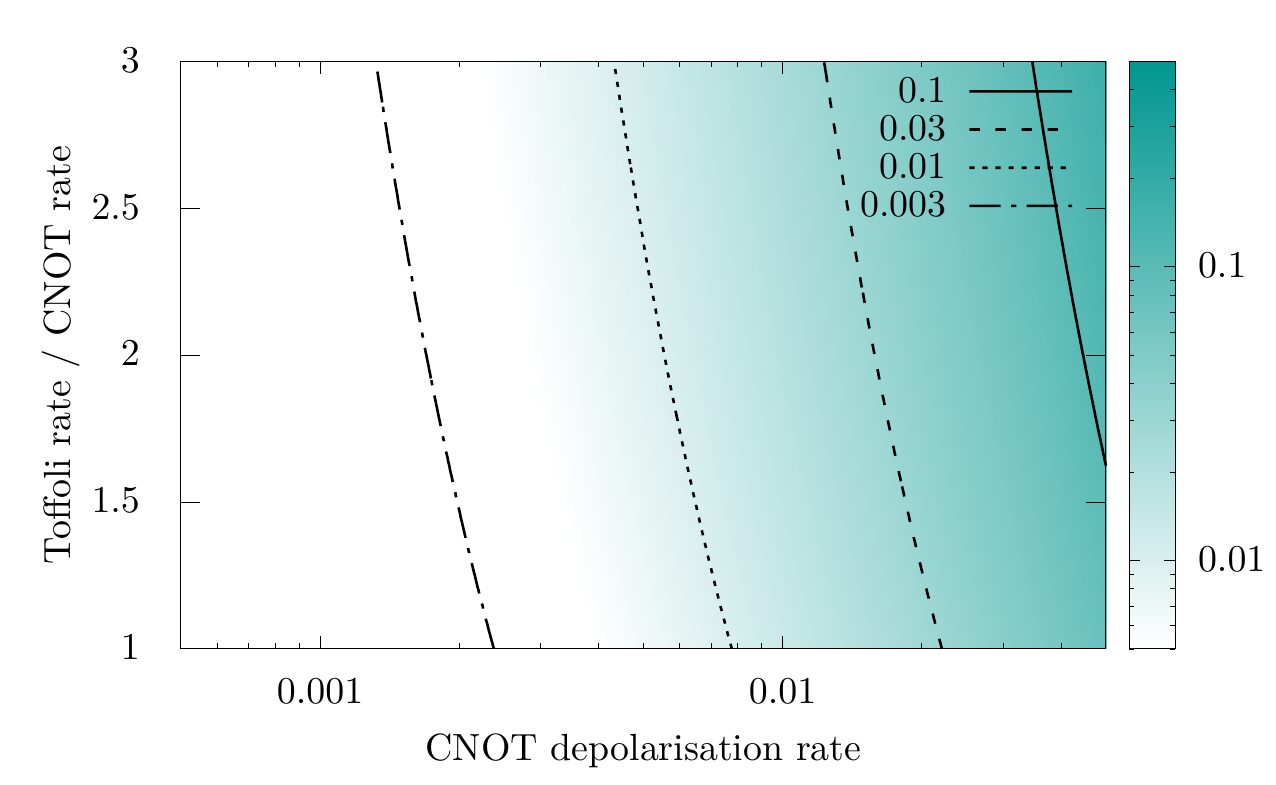}
\caption{Gate thresholds for the $(3,1,1)$ circuit.}
\label{fig:threshold-3}
\end{figure}

In the rest of this paper we propose and assess the performance of a range of circuits analogous to that of Figure~\ref{fig:(3,1,1)}, and quantify the trade-off in which more complicated circuits provide greater protection against preparation errors but are more vulnerable to gate errors.
In Section~\ref{sec:purification} we define and describe the general properties of purification circuits, explain how to find small examples computationally, and consider ways to combine these small examples to form larger circuits.
In Section~\ref{sec:graphs} we describe a particular construction based on graphs inspired by examples discovered in Section~\ref{sec:purification}.
The simplest versions of these circuits have favourable performance and resource requirements and have a natural planar layout.
In Section~\ref{sec:performance} we present data on the performance of these circuits.

Various schemes making repeated applications of the idea behind Figure~\ref{fig:measurement-based} have been studied under the name `algorithmic cooling'~\cite{SchulmanVazirani99,boykin2002algorithmic,practicable,PhysRevA.83.042340,brassard2014prospects,laflamme2022algorithmic}.
Their analysis is typically in terms of entropy, tracking which qubits are `hot' or `cold'.
Our approach is combinatorial, instead tracking what happens to individual error patterns as the circuit permutes the computational basis.

\section{Purification circuits}
\label{sec:purification}

We define an $(n, k, e)$ \defn{purification circuit} (or $(n,k,e)$ \defn{circuit}) to be a CNOT and Toffoli circuit mapping the $n$-long bitstrings of weight at most $e$ into the space of $n$-long bitstrings that are zero in $k$ nominated positions.
If each value in the input is intended to be $0$, but experiences an independent error probability $\prep$ of being flipped to a $1$, then each of the $k$ outputs of an $(n, k, e)$ purification circuit experiences a reduced error rate $O(\prep^{e+1})$. 

The parameter $e$ is defined in terms of protection against adversarial errors and, as with the distance of a (quantum or classical) error correcting code, does not necessarily give a complete picture of the degree of protection against random errors.
The worst case output error rate of an $(n,k,e)$ purification circuit is $\binom n {e+1} \prep^{e+1} + O(\prep^{e+2})$.
In many cases only a small fraction of the sets of $e+1$ preparation errors lead to an error on the output, so the coefficient of $\prep^{e+1}$ will be much smaller.
This applies, for example, to the family of circuits that we consider in Section~\ref{sec:graphs}.

\subsection{Error model}
\label{sec:error-model}

An $(n,k,e)$ purification circuit can be viewed as a purely classical object operating on classical bitstrings.
To interpret it as a quantum circuit running on noisy hardware we adopt the following error model for preparation, idle and gate errors.

The basic error events are:
\begin{itemize}
\item 
a qubit is incorrectly prepared as $\ket 1$ rather than $\ket 0$ with probability $\prep$;

\item
a qubit not involved in the current round of gates depolarises with probability $\idle$;

\item
a CNOT gate depolarises the set of qubits it acts on with probability $\cnot$;

\item
a Toffoli gate depolarises the set of qubits it acts on with probability $\toffoli$.
\end{itemize}
All of these events are independent.

By `a set $Q$ of qubits depolarises' we mean any of the following equivalent things.
\begin{itemize}
\item 
the qubits in $Q$ are replaced by a uniform mixture of the computational basis;

\item
each qubit in $Q$ experiences a Pauli $I$, $X$, $Y$ or $Z$ error chosen uniformly and independently at random;

\item
each qubit in $Q$ experiences a Pauli $X$ error with probability $1/2$ and a Pauli $Z$ error with probability $1/2$, with all choices made independently.
\end{itemize}

Note that the parameterisation for preparation errors differs from that for idle and gate errors; preparing $\ket 1$ rather than $\ket 0$ with probability $\prep<1/2$ corresponds to correctly preparing $\ket 0$ then depolarising with probability $2\prep$.

Since failures in either preparation or application of a gate in this model replace qubits by mixtures of computational basis states, the state of the system at any point is fully described by a probability distribution over the computational basis.
Each probability is a polynomial in $\prep, \idle, \cnot, \toffoli$, which can be computed precisely for circuits of moderate size.
See Appendix~\ref{app:assessing} for more details.

This simple model has the significant advantage of being easy to compute with.
The disadvantage is that it does not capture all possible error processes within a quantum computer.
To give just one example, suppose that your CNOT gate is composed of CZ and Hadamard gates, and that a more accurate error model is that each component gate has an independent probability of depolarising the qubits it acts on.
This is not the same as the combined CNOT gate having some probability of depolarising the qubits it acts on; it doesn't even have the property that a system with this noise model can be described as a mixture of computational basis states.
Fully realistic noise models are of course more complicated again.

\subsection{Existence}
\label{sec:existence}

We can view an $(n,k,e)$ purification circuit as a permutation of $\mathbb F_2^n$ which maps the set of vectors of weight at most $e$ into the set of vectors that vanish in $k$ nominated positions.
This places a size constraint on $n, k, e$.
The necessary condition is also sufficient.

\begin{proposition}\label{prop:existence}
Let $n \geq 1$ and $k, e \in \{1, \ldots, n\}$.
If 
\begin{equation}\label{eqn:space-bound}
    \binom n 0 + \binom n 1 + \cdots + \binom n e \leq 2^{n-k},
\end{equation}
then there is an $(n,k,e)$ purification circuit consisting of  CNOT and Toffoli gates.
\end{proposition}

Note that this is false if we restrict to circuits containing only CNOTs (which act linearly on $\mathbb F_2^n$) or to circuits containing only Toffolis (which fix the set of states of weight at most $1$).
Any purification circuit for which \eqref{eqn:space-bound} is equality is optimal in the sense that, for $\prep < 1/2$, it maps the most likely $2^{n-k}$ basis states to the most useful $2^{n-k}$ states.

We prove Proposition~\ref{prop:existence} in Appendix~\ref{app:existence}.
The argument is group-theoretic, and does not provide an efficient procedure to construct purification circuits with given parameters.
For the special case of $(2^{m+1}-1, 1, 2^m-1)$ circuits we describe an explicit if impractical construction in Appendix~\ref{app:explicit}.

\subsection{Finding small circuits}
\label{sec:small-circuits}

\begin{figure}
\centering
\subcaptionbox{\label{fig:(5,1,2)a}}
{
\begin{tikzpicture}
  \begin{yquant}
    qubit {$q_{\idx}$} q[5];
    cnot q[3] | q[0], q[1];
    cnot q[4] | q[0], q[2];
    cnot q[1] | q[0];
    cnot q[2] | q[0];
    cnot q[0] | q[1], q[2];
    cnot q[0] | q[3], q[4];
    cnot q[3] | q[0], q[1];
    cnot q[4] | q[0], q[2];
    cnot q[0] | q[3], q[4];
  \end{yquant}
\end{tikzpicture}
}
\subcaptionbox{\label{fig:(5,1,2)b}}
{
\begin{tikzpicture}
  \begin{yquant}
    qubit {$q_{\idx}$} q[5];
    cnot q[3] | q[1], q[2];
    cnot q[1] | q[0];
    cnot q[2] | q[0];
    cnot q[0] | q[2], q[1];
    cnot q[4] | q[0];
    cnot q[0] | q[3], q[2];
    cnot q[2] | q[1], q[4];
    cnot q[3] | q[0], q[4];
    cnot q[0] | q[3], q[2];
  \end{yquant}
\end{tikzpicture}
}
\caption{Examples of the two classes of $(5,1,2)$ purification circuits of length $9$.}\label{fig:(5,1,2)}
\end{figure}
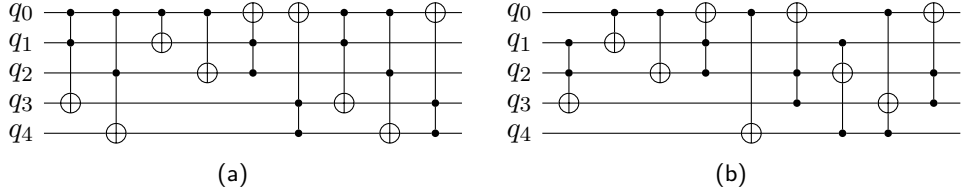

An $(n, 1, 1)$ purification circuit must have $n \geq 3$.
It is straightforward to check that no one or two gate circuit is a $(3, 1, 1)$ circuit, so the $(3,1,1)$ circuit in Figure~\ref{fig:(3,1,1)} is the smallest example that can defend against a single preparation error.
(There is one other family of examples, obtained by replacing the second CNOT gate by a Toffoli.)
Similarly, an $(n, 1, 2)$ circuit must have $n \geq 5$.
An exhaustive search reveals that the smallest $(5,1,2)$ circuits have $9$ gates; two examples are shown in Figure~\ref{fig:(5,1,2)}.
Equivalent circuits can be obtained by permuting the bottom four wires, or by interchanging consecutive commuting gates.
All 384 $(5,1,2)$ circuits of length $9$ which output on the top wire can be obtained from these examples in this way: 288 variations of the first type and 96 variations of the second.

There are $2^5 = 32$ basis states of 5 qubits, so a set of states (such as the low weight states, or the states that are $\ket 0$ on the first qubit) can be represented by a string of 32 bits.
By working forward from the initial set of states, backward from the target set of states and meeting in the middle, we can find $(5, 1, 2)$ circuits very quickly, in around $\sqrt{2^{32}} = 2^{16}$ time and space; see Appendix~\ref{app:searching} for more details.

To find a $(7, 1, 3)$ circuit in this way the comparable figure is $\sqrt{2^{2^7}} = 2^{64}$ time and space, so we are already at the limit of what can be achieved without special insight into the problem.
In the next section we present techniques for combining purification circuits.
The results do not have optimal parameters $(n, k, e)$, but do have relatively simple structures and avoid the requirement to do large amounts of work upfront to discover them.

\subsection{New purification circuits from old}\label{sec:new-from-old}

In this section we describe two techniques for constructing new purification circuits from existing ones.

\subsubsection{Composition}

\begin{figure}
\centering
\begin{tikzpicture}
  \begin{yquant}
    qubit {$q_{\idx}$} q[9];
    cnot q[7] | q[6];
    cnot q[4] | q[3];
    cnot q[1] | q[0];
    cnot q[8] | q[6];
    cnot q[5] | q[3];
    cnot q[2] | q[0];
    cnot q[6] | q[7], q[8];
    cnot q[3] | q[4], q[5];
    cnot q[0] | q[1], q[2];
    cnot q[3] | q[0];
    cnot q[6] | q[0];
    cnot q[0] | q[3], q[6];
  \end{yquant}
\end{tikzpicture}
\caption{Composing $(3,1,1)$ circuits to obtain a $(9, 1, 3)$ circuit.}
\label{fig:composition}
\end{figure}
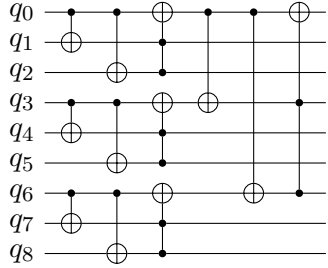

The simplest and most general technique is to compose circuits, of similar or different types, to obtain larger ones.
Figure~\ref{fig:composition} shows three $(3,1,1)$ circuits feeding their outputs into a fourth $(3,1,1)$ circuit.
An error on the output requires errors on at least two of the inputs of the fourth circuit, which requires at least two errors on at least two of the first three circuits, so the composed circuit has parameters $(9, 1, 3)$.
By Proposition~\ref{prop:existence}, a $(9,1,e)$ purification circuit could in principle have $e = 4$, so we have given up something on the achievable protection in exchange for a concrete circuit with a simple structure.

By the same argument we obtain the following.

\begin{proposition}\label{prop:composition}
Feeding the outputs of $n_2$ copies of an $(n_1, 1, e_1)$ purification circuit with $g_1$ gates into an $(n_2, 1, e_2)$ purification circuit with $g_2$ gates produces an $(n_1n_2, 1, (e_1+1)(e_2+1)-1)$ circuit with $n_2g_1 + g_2$ gates.\qed
\end{proposition}

Observe that feeding $(5, 1, 2)$ circuits into a $(3, 1, 1)$ circuit or feeding $(3, 1, 1)$ circuits into a $(5, 1, 2)$ both result in $(15, 1, 5)$ circuits, but $5$-into-$3$ uses 30 gates and $3$-into-$5$ uses 20 gates, so there is typically an incentive to use simpler circuits closer to the raw input.

\begin{figure}
\centering
  \subcaptionbox{$(5,1,1)$\label{fig:(5,1,1)}}
{
\begin{tikzpicture}
  \begin{yquant}
    qubit {$q_{\idx}$} q[5];
    cnot q[1] | q[0];
    cnot q[2] | q[0];
    cnot q[0] | q[1], q[2];
    cnot q[3] | q[0];
    cnot q[4] | q[0];
    cnot q[0] | q[3], q[4];
  \end{yquant}
\end{tikzpicture}
}
  \subcaptionbox{$(7,1,2)$\label{fig:(7,1,2)}}
{
\begin{tikzpicture}
  \begin{yquant}
    qubit {$q_{\idx}$} q[7];
    cnot q[4] | q[3];
    cnot q[1] | q[0];
    cnot q[5] | q[3];
    cnot q[2] | q[0];
    cnot q[3] | q[4], q[5];
    cnot q[0] | q[1], q[2];
    cnot q[3] | q[0];
    cnot q[6] | q[0];
    cnot q[0] | q[3], q[6];
  \end{yquant}
\end{tikzpicture}
}
\caption{Interpolating between $(3,1,1)$ and $(9,1,3)$.}
\label{fig:interpolation}
\end{figure}
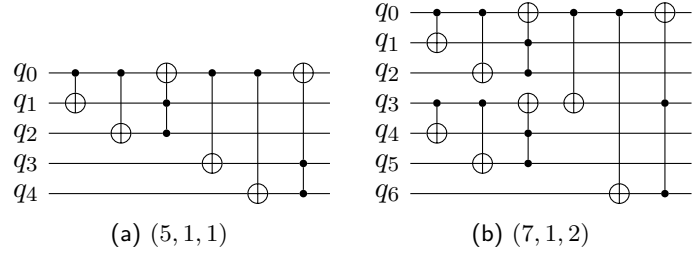

There is no requirement that the circuits in the first stage are identical.
For example, we might replace some of the first round $(3,1,1)$ circuits in the $(9,1,3)$ circuit by naive preparations.
Changing only some of the inputs in this way, as shown in Figure~\ref{fig:interpolation}, allows us to interpolate between the performance and resource requirements of the $(3,1,1)$ circuit and the $(9,1,3)$ circuit.

\subsubsection{Juxtaposition and overlapping}\label{sec:juxtaposition}

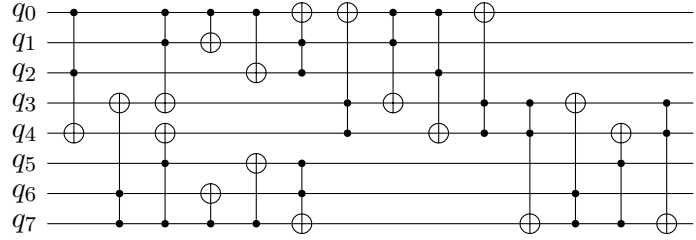
\begin{figure}
\centering
\begin{tikzpicture}
  \begin{yquant}
    qubit {$q_{\idx}$} q[8];
    cnot q[4] | q[0], q[2];
    cnot q[3] | q[7], q[6];
    cnot q[3] | q[0], q[1];
    cnot q[4] | q[7], q[5];
    cnot q[6] | q[7];
    cnot q[1] | q[0];
    cnot q[5] | q[7];
    cnot q[2] | q[0];
    cnot q[7] | q[6], q[5];
    cnot q[0] | q[1], q[2];
    cnot q[0] | q[3], q[4];
    cnot q[3] | q[0], q[1];
    cnot q[4] | q[0], q[2];
    cnot q[0] | q[3], q[4];
    cnot q[7] | q[3], q[4];
    cnot q[3] | q[7], q[6];
    cnot q[4] | q[7], q[5];
    cnot q[7] | q[3], q[4];
  \end{yquant}
\end{tikzpicture}
\caption{An $(8,2,2)$ circuit obtained by overlapping two $(5,1,2)$ circuits.  Output is on the first and last wires.}
\label{fig:(8,2,2)}
\end{figure}

So far all of our circuits have produced a single output.
The simplest way to get multiple outputs is to use multiple circuits; $m$ copies of an $(n, k, e)$ circuit can be viewed as a single $(mn, mk, e)$ circuit.
This is very far from optimal, as we are over-protected against sets of $e$ errors split between the $m$ circuits.
In favourable situations we can exploit this fact by overlapping circuits to re-use qubits.
For example, the circuit in Figure~\ref{fig:(8,2,2)} comprises two copies of the $(5,1,2)$ circuit from Figure~\ref{fig:(5,1,2)a} sharing qubits $q_3$ and $q_4$.
We can check that this circuit retains the full protection of both original copies, for an overall parameter set $(8,2,2)$.
This circuit has been drawn to highlight its symmetry, but its depth can be reduced by $1$ to $13$ by performing the first four gates over two rounds.

\subsection{Relation to classical codes}

\begin{figure}
\centering
\subcaptionbox{A legible ordering of the gates.\label{fig:hamming-legible}}
{
\begin{tikzpicture}
  \begin{yquant}
    qubit {$q_\idx$} q[7];
    cnot q[3] | q[0];
    cnot q[6] | q[0];
    cnot q[0] | q[3], q[6];
    cnot q[4] | q[1];
    cnot q[6] | q[1];
    cnot q[1] | q[4], q[6];
    cnot q[5] | q[2];
    cnot q[6] | q[2];
    cnot q[2] | q[5], q[6];
    cnot q[4] | q[3];
    cnot q[5] | q[3];
    cnot q[3] | q[4], q[5];
  \end{yquant}
\end{tikzpicture}
}
\subcaptionbox{A logically equivalent circuit that can be scheduled over 8 rounds.\label{fig:hamming-scheduled}}
{
\begin{tikzpicture}
  \begin{yquant}
    qubit {$q_{\idx}$} q[7];
    cnot q[5] | q[2];
    cnot q[4] | q[1];
    cnot q[3] | q[0];
    cnot q[4] | q[3];
    cnot q[6] | q[0];
    cnot q[5] | q[3];
    cnot q[6] | q[1];
    cnot q[6] | q[2];
    cnot q[0] | q[3], q[6];
    cnot q[1] | q[4], q[6];
    cnot q[2] | q[5], q[6];
    cnot q[3] | q[4], q[5];
  \end{yquant}
\end{tikzpicture}
}
\caption{$(7,4,1)$ purification circuits based on the Hamming code.}
\label{fig:hamming}
\end{figure}
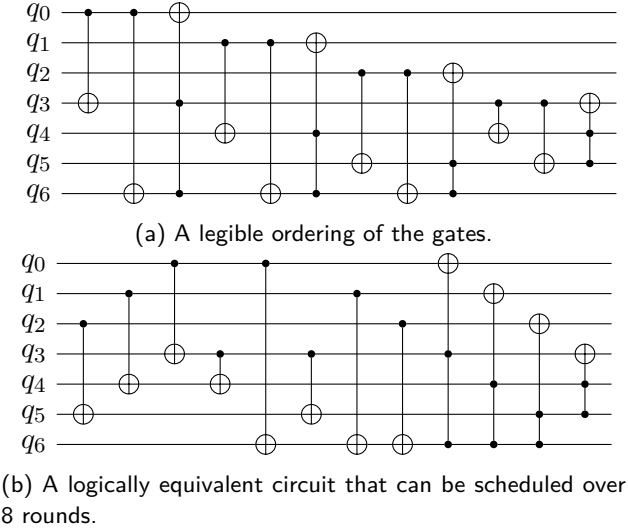

A $(2e+1, 1, e)$ purification circuit can be viewed as a decoder for the $(2e+1)$-bit classical repetition code.
The output bit contains the majority vote of the input bits, and the other bits contain the syndrome information required to reconstruct the input.
Any CNOT and Toffoli circuit for decoding a classical $[n,k,2e+1]$ code can be viewed as an $(n,k,e)$ purification circuit in the same way.

This is generally more than we need; we only require that strings close to the zero codeword are correctly decoded.
Take for example the $[7,4,3]$ Hamming code.
By Proposition~\ref{prop:existence}, there is a CNOT and Toffoli circuit bijecting the ball of radius 1 about 0 with the space of vectors that vanish in $4$ nominated positions.
Since there are only $\binom{128}8 \approx 2^{40}$ sets of 8 basis states on 7 qubits, searching for these circuits is feasible, unlike the situation for $(7,1,3)$ circuits described in Section~\ref{sec:small-circuits}.
There are $508\,022\,784$ such circuits.
By inspecting a small number of these and rearranging gates by hand we arrive at the circuit presented in Figure~\ref{fig:hamming-legible}.
We emphasise that this is not a full decoding circuit for the Hamming code.

We can interpret this circuit as four consecutive $(3,1,1)$ circuits on sets of wires $\{0,3,6\}$, $\{1,4,6\}$, $\{2,5,6\}$, $\{3,4,5\}$.
The first three circuits share the auxiliary qubit $q_6$.
The final circuit re-uses the remaining auxiliary qubits from the first three circuits.
Figure~\ref{fig:hamming-scheduled} shows a logically equivalent circuit obtained by commuting gates past each other greedily that can be scheduled over 8 rounds.
This is the form of the circuit that we use in simulations.

In the next section we present a very general construction inspired by circuits like that in Figure~\ref{fig:hamming-legible}.

\section{Purification circuits from graphs}
\label{sec:graphs}

With perfect gates, an $(n,k,e)$ purification circuit improves a preparation error rate $\prep$ to $O(\prep^{e+1})$.
In practice purification circuits will themselves be subject to error, meaning that the dominant term in the output error probability is likely to be due to gate failures, for example of the final Toffoli.
There is therefore limited advantage to increasing $e$.
In this section we focus on increasing the \defn{rate} $k/n$, presenting a general construction, based on graphs, which allows us to tune the rate against other parameters like circuit depth.

\subsection{Reinterpreting the $(3,1,1)$ circuit}

We can interpret the operation of the $(3,1,1)$ circuit as follows.
Think of the first wire as a data qubit, and the other two wires as auxiliary qubits.
When the single allowed error is on the data, the two CNOTs mark the auxiliary qubits.
This double mark is detected by the Toffoli, which clears the data error.
When the error is on an auxiliary qubit, it can't spread back through the CNOTs, and it isn't able to activate both controls of the Toffoli, so the data remains unchanged.
If instead we have $a$ auxiliary qubits then we have $\binom a 2$ pairs available for marking in this way, which can be used to protect up to $\binom a 2$ data qubits against a single preparation error.
This construction is naturally expressed in the language of graphs.

\subsection{The general construction}

Given a graph $G$ with $r$ vertices and $s$ edges, we can obtain an $(r+s, s, 1)$ purification circuit as follows.
Associate one data qubit $q_{uv}$ to each edge $uv$, and one auxiliary qubit $q_v$ to each vertex $v$.
\begin{enumerate}[leftmargin=8em]
    \item[(detect stage)] For each edge $uv$, apply CNOTs controlled on $q_{uv}$ and targeting $q_u$, $q_v$.
    \item[(correct stage)] For each edge $uv$, apply Toffolis controlled on $q_u$, $q_v$ and targeting $q_{uv}$.
\end{enumerate}
The gates within each stage commute, so can be performed in any order.
The circuit has $2s$ CNOTs and $s$ Toffolis, and can be scheduled over at most $3\chi_e(G)$ rounds, where $\chi_e(G)$ is the edge-chromatic number of $G$.
See Figure~\ref{fig:graph} for the result of applying this process to a path of length $5$.

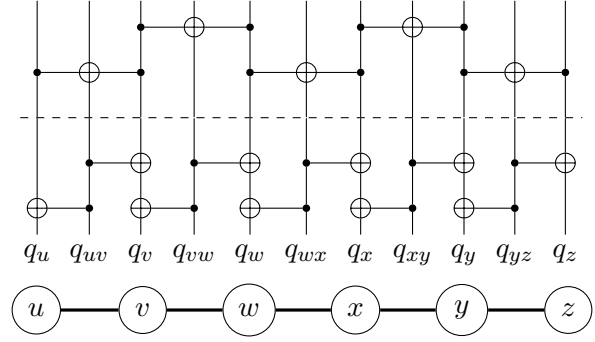
\begin{figure}
\centering
\begin{tikzpicture}[rotate=90]
  \begin{yquant}[register/separation=7pt]
    qubit {\rotatebox{-90}{$q_u$}} u;
    qubit {\rotatebox{-90}{$q_{uv}$}} uv;
    qubit {\rotatebox{-90}{$q_v$}} v;
    qubit {\rotatebox{-90}{$q_{vw}$}} vw;
    qubit {\rotatebox{-90}{$q_w$}} w;
    qubit {\rotatebox{-90}{$q_{wx}$}} wx;
    qubit {\rotatebox{-90}{$q_x$}} x;
    qubit {\rotatebox{-90}{$q_{xy}$}} xy;
    qubit {\rotatebox{-90}{$q_y$}} y;
    qubit {\rotatebox{-90}{$q_{yz}$}} yz;
    qubit {\rotatebox{-90}{$q_z$}} z;
    cnot u | uv;
    cnot v | vw;
    cnot w | wx;
    cnot x | xy;
    cnot y | yz;
    cnot v | uv;
    cnot w | vw;
    cnot x | wx;
    cnot y | xy;
    cnot z | yz;
    barrier (u,uv,v,vw,w,wx,x,xy,y,yz,z);
    cnot uv | u, v;
    cnot wx | w, x;
    cnot yz | y, z;
    cnot vw | v, w;
    cnot xy | x, y;
  \end{yquant}
  \tikzmath{\step = 1.4;\offset=-0.17;}
  \node[circle,draw=black] (u) at (-1,\offset) {$u$};
  \node[circle,draw=black] (v) at (-1,\offset-\step) {$v$};
  \node[circle,draw=black] (w) at (-1,\offset-\step*2) {$w$};
  \node[circle,draw=black] (x) at (-1,\offset-\step*3) {$x$};
  \node[circle,draw=black] (y) at (-1,\offset-\step*4) {$y$};
  \node[circle,draw=black] (z) at (-1,\offset-\step*5) {$z$};
  \draw[very thick] (u) -- (v) -- (w) -- (x) --(y) -- (z);
\end{tikzpicture}
\caption{In the graph construction, each edge becomes a $(3,1,1)$ circuit.  In the first stage, CNOTs copy an error on any edge to its endpoints.  In the second stage, Toffolis correct the error on any edge which marked its endpoints in this way.  The gates within each stage commute and can be applied in any order.}\label{fig:graph}
\end{figure}

\begin{figure}
    \centering
    \includegraphics[page=1,scale=0.5]{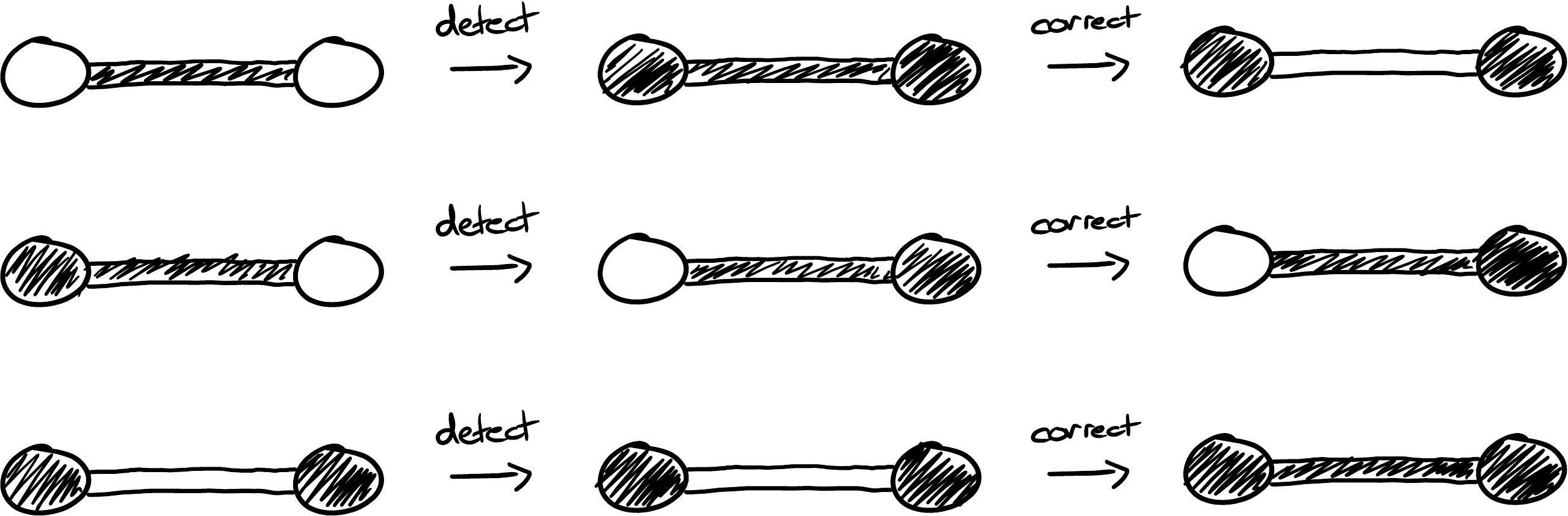}
    \caption{Correct operation of the $(3,1,1)$ circuit viewed as a graph and the two different ways in which it fails to handle two errors.}
    \label{fig:edge-graph}
\end{figure}

There is a superficial resemblance between these circuits and those implementing a surface code (or repetition code, in the case where $G$ is a path or cycle).
Data qubits are associated with each edge, and `syndrome extraction' consists of CNOTs implementing a boundary operator, revealing those vertices incident to an odd number of edges with data errors.
The two processes then diverge, as a true surface code makes corrections based on considering a global syndrome; we instead make local corrections to edges both of whose endpoints are marked in the syndrome.
This works well for isolated errors, but not for more complicated error configurations.

\begin{figure}
\centering
\subcaptionbox{\label{fig:k3}}{\includegraphics[page=3,scale=0.7]{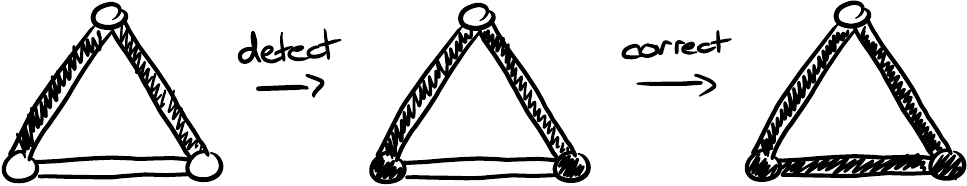}}\\[1em]
\subcaptionbox{\label{fig:k4}}{\includegraphics[page=4,scale=0.55]{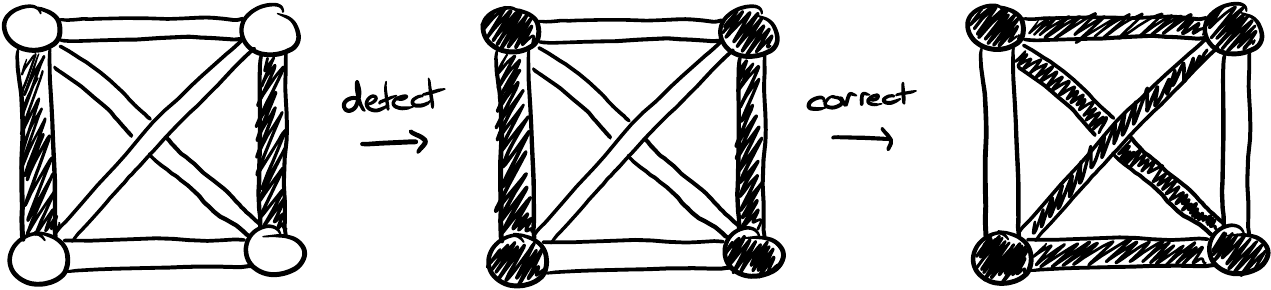}}
\caption{Failure of the complete graph construction with data errors on two different edges.}
\label{fig:2-failures}
\end{figure}

Figure~\ref{fig:edge-graph} shows the correct operation of the $(3,1,1)$ circuit when there is a single error on the data edge, and the two ways it can fail when there are two errors spread across the data edge and auxiliary vertices.
In larger graphs there are other modes of failure.
Errors on two incident edges cancel out part of their boundary, leading to neither error being cleared and inducing an error on any edge spanning their other two endpoints (Figure~\ref{fig:k3}).
Errors on two disjoint edges are successfully cleared by the circuit, but cause errors on any other edges induced by these vertices (Figure~\ref{fig:k4}).

Taking $G$ to be the complete graph $K_r$ produces an $(r+\binom r 2, \binom r 2, 1)$ purification circuit. 
If we write $k = \binom r 2$, then the parameters become $(k+O(\sqrt k), k, 1)$; that is, we can protect a set of $k$ qubits against a single preparation error with only $O(\sqrt k)$ overhead.  The optimal overhead is at least $\log_2 k$ (as we must be able to map every single qubit error to an error pattern supported on the auxiliary qubits), so this is not too far from best possible.

\subsection{Paths and cycles}

In addition to failing for many sets of two preparation errors (Figure~\ref{fig:2-failures}), the circuit for $K_r$ has the disadvantage that it requires $\Omega(r)$ rounds of gates.
These disadvantages can be addressed by choosing a sparser graph, such a long cycle.
If $G$ is a cycle $C_k$ on $k$ vertices, then the circuit can be scheduled over $4$ (if $k$ is even) or $5$ (if $k$ is odd) rounds, and has a natural planar layout (cf.\ Figure~\ref{fig:graph}).
It also significantly outperforms its parameters.

\begin{figure}
    \centering
  \subcaptionbox{Output on $q_3$.\label{fig:small-light-cone}}
{
\begin{tikzpicture}
  \begin{yquant}
    qubit {$q_{\idx}$} q[6];
    cnot q[4] | q[5];
    cnot q[0] | q[1];
    cnot q[2] | q[3];
    cnot q[2] | q[1];
    cnot q[4] | q[3];
    cnot q[3] | q[2], q[4];
  \end{yquant}
\end{tikzpicture}
}
  \subcaptionbox{Output on $q_5$.\label{fig:large-light-cone}}
{
\begin{tikzpicture}
  \begin{yquant}
    qubit {$q_{\idx}$} q[10];
    cnot q[8] | q[9];
    cnot q[4] | q[5];
    cnot q[0] | q[1];
    cnot q[6] | q[7];
    cnot q[2] | q[3];
    cnot q[6] | q[5];
    cnot q[2] | q[1];
    cnot q[8] | q[7];
    cnot q[4] | q[3];
    cnot q[7] | q[6], q[8];
    cnot q[3] | q[2], q[4];
    cnot q[5] | q[4], q[6];
  \end{yquant}
\end{tikzpicture}
}
    \caption{Small and large light cones over a long even cycle.  When viewed as parts of the same circuit over a path or cycle, there is a further idle step following the short light cone.}
    \label{fig:light-cone}
\end{figure}
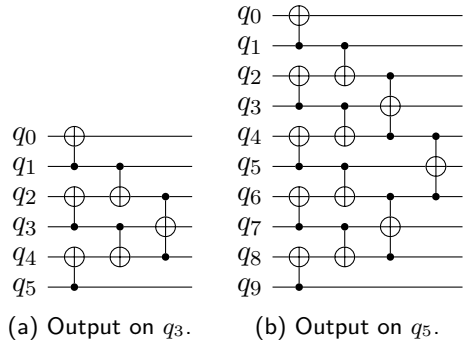

A $(2k,k,1)$ circuit defined over a cycle protects completely against any single error.
Two or more errors on the input might lead to an error on the output, but only if those errors are sufficiently close.
By the \defn{light cone} of an output qubit we mean the subset of the gates and input qubits that can causally affect it.
The state of the output qubit depends only on the preparation and gate errors on this part of the circuit.
Over a long even cycle there are only two types of output qubit up to isomorphism (those in the first round of Toffolis and those in the second round) and so two possible light cones, shown in Figure~\ref{fig:light-cone}.
This allows us to analyse the output quality of the circuit over any sufficiently long even cycle by examining only two light cones.
Since the larger light cone only reaches 10 qubits, `sufficiently long' means $k \geq 10$.

With preparation error rate $\prep$ and perfect gates, each output qubit experiences errors with probability $O(\prep^2)$.
Errors on output qubits are independent if they are sufficiently separated that their light cones are disjoint, but errors on nearby output qubits are correlated.
This might lead to the following situation.
Suppose that some set of $a$ errors on the input leads to $b > a$ errors on the output.
Then an event which should have probability $O(\prep^b)$ in fact occurs with probability $\Omega(\prep^a)$.
Depending on the intended use of the output qubits, this might be problematic.

\begin{figure}
    \centering
    \includegraphics[page=2,scale=0.7]{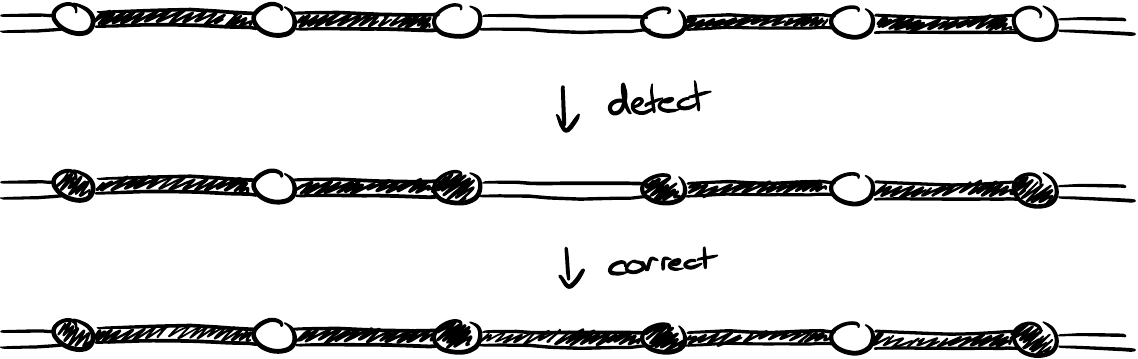}
    \caption{Smallest non-fault tolerant error pattern on a path or long cycle.}
    \label{fig:path-graph}
\end{figure}

We call a purification circuit \defn{combinatorially fault-tolerant up to $b$ preparation errors} if, for any $a \leq b$, any set of $a$ errors on the input leads to at most $a$ errors on the output.
For example, one can check that the circuits corresponding to long cycles are combinatorially fault-tolerant up to $3$ preparation errors and that there is exactly one pattern of $4$ errors on the input leading to more than $4$ errors on the output (Figure~\ref{fig:path-graph}).
This pattern works as follows.
Input errors on two adjacent edges are frozen in place as their shared auxiliary qubit is first set to $\ket 1$ then reset to $\ket 0$, so does not activate the Toffolis.
Two of these frozen patterns placed one edge apart preserve the original errors but also causes a new error on the separating edge.
            
\begin{figure}
\centering
\begin{tikzpicture}[rotate=90]
  \begin{yquant}[register/separation=7pt]
    qubit {\rotatebox{-90}{$q_u$}} u;
    qubit {\rotatebox{-90}{$q_{uv}$}} uv;
    qubit {\rotatebox{-90}{$q_v$}} v;
    qubit {\rotatebox{-90}{$q_{vw}$}} vw;
    qubit {\rotatebox{-90}{$q_w$}} w;
    qubit {\rotatebox{-90}{$q_{wx}$}} wx;
    qubit {\rotatebox{-90}{$q_x$}} x;
    qubit {\rotatebox{-90}{$q_{xy}$}} xy;
    qubit {\rotatebox{-90}{$q_y$}} y;
    qubit {\rotatebox{-90}{$q_{yz}$}} yz;
    qubit {\rotatebox{-90}{$q_z$}} z;
    cnot u | uv;
    cnot v | vw;
    cnot w | wx;
    cnot x | xy;
    cnot y | yz;
    cnot v | uv;
    cnot w | vw;
    cnot x | wx;
    cnot y | xy;
    cnot z | yz;
    cnot v | uv, vw;
    cnot x | wx, xy;
    cnot w | vw, wx;
    cnot y | xy, yz;
    barrier (u,uv,v,vw,w,wx,x,xy,y,yz,z);
    cnot uv | u, v;
    cnot wx | w, x;
    cnot yz | y, z;
    cnot vw | v, w;
    cnot xy | x, y;
  \end{yquant}
  \tikzmath{\step = 1.4;\offset=-0.17;}
  \node[circle,draw=black] (u) at (-1,\offset) {$u$};
  \node[circle,draw=black] (v) at (-1,\offset-\step) {$v$};
  \node[circle,draw=black] (w) at (-1,\offset-\step*2) {$w$};
  \node[circle,draw=black] (x) at (-1,\offset-\step*3) {$x$};
  \node[circle,draw=black] (y) at (-1,\offset-\step*4) {$y$};
  \node[circle,draw=black] (z) at (-1,\offset-\step*5) {$z$};
  \draw[very thick] (u) -- (v) -- (w) -- (x) --(y) -- (z);
\end{tikzpicture}
\caption{Graph circuits arising from paths or cycles can be made fully fault-tolerant by adding additional Toffolis to ensure that the shared vertex is set to $\ket 1$ when there are errors on consecutive edges.}\label{fig:tolerant-graph}
\end{figure}
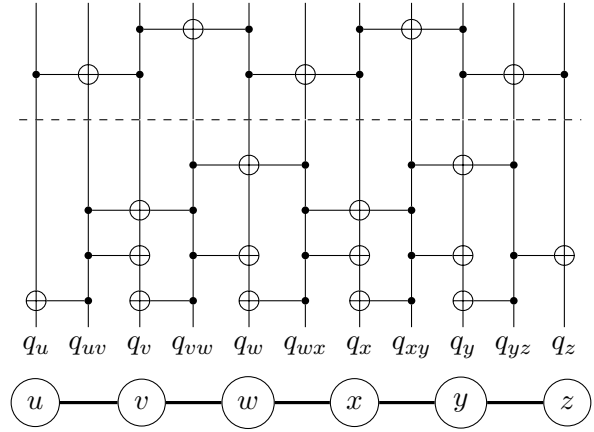
\begin{figure}
  \centering
  \subcaptionbox{Output on $q_5$.\label{fig:small-tolerant-light-cone}}
{
\begin{tikzpicture}
  \begin{yquant}
    qubit {$q_{\idx}$} q[10];
    cnot q[6] | q[7];
    cnot q[2] | q[3];
    cnot q[8] | q[9];
    cnot q[4] | q[5];
    cnot q[0] | q[1];
    cnot q[8] | q[7];
    cnot q[4] | q[3];
    cnot q[6] | q[5];
    cnot q[2] | q[1];
    cnot q[6] | q[5], q[7];
    cnot q[2] | q[1], q[3];
    cnot q[4] | q[3], q[5];
    cnot q[5] | q[4], q[6];
  \end{yquant}
\end{tikzpicture}
}
  \subcaptionbox{Output on $q_7$.\label{fig:large-tolerant-light-cone}}
{
\begin{tikzpicture}
  \begin{yquant}
    qubit {$q_{\idx}$} q[14];
    cnot q[10] | q[11];
    cnot q[6] | q[7];
    cnot q[2] | q[3];
    cnot q[12] | q[13];
    cnot q[8] | q[9];
    cnot q[4] | q[5];
    cnot q[0] | q[1];
    cnot q[12] | q[11];
    cnot q[8] | q[7];
    cnot q[4] | q[3];
    cnot q[10] | q[9];
    cnot q[6] | q[5];
    cnot q[2] | q[1];
    cnot q[10] | q[9], q[11];
    cnot q[6] | q[5], q[7];
    cnot q[2] | q[1], q[3];
    cnot q[8] | q[7], q[9];
    cnot q[4] | q[3], q[5];
    cnot q[9] | q[8], q[10];
    cnot q[5] | q[4], q[6];
    cnot q[7] | q[6], q[8];
  \end{yquant}
\end{tikzpicture}
}
    \caption{Small and large light cones for the fault-tolerant circuit over a long even cycle.  When viewed as parts of the same circuit over a path or cycle, there is a further idle step following the short light cone.}
    \label{fig:tolerant-light-cone}
\end{figure}
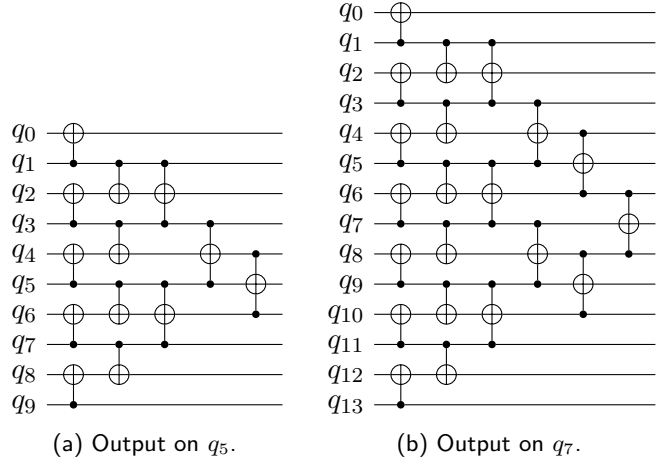

One way to stop this behaviour is to prevent errors on adjacent pairs of edges from being frozen in place.
We can do this by adding an extra round of Toffolis to the detect stage.
\begin{enumerate}[leftmargin=8em]
    \item[(detect stage)] For each edge $uv$, apply CNOTs controlled on $q_{uv}$ and targeting $q_u$, $q_v$.
    \item[(detect$'$ stage)] For each pair of consecutive edges $uv$, $vw$, apply Toffolis controlled on $q_{uv}$, $q_{vw}$ and targeting $q_v$.
    \item[(correct stage)] For each edge $uv$, apply Toffolis controlled on $q_u$, $q_v$ and targeting $q_{uv}$.
\end{enumerate}
The gates in both detect stages all commute with each other, so can be performed in any order.
One choice of circuit is shown in Figure~\ref{fig:tolerant-graph}, with corresponding light cones in Figure~\ref{fig:tolerant-light-cone}.

The idea behind this extended circuit is that, if there are no input errors on the vertices, then the detect stage marks \emph{all} of the endpoints of each edge with an input error, with no cancellation arising from errors on consecutive edges.
Then an input error on an edge is successfully cleared unless there is also an input error on one of its endpoints;
and an edge with no input error experiences an error on the output only if, for both of its endpoints, there is an input error on either that vertex or the next edge.

\begin{theorem}\label{thm:tolerant}
The extended purification circuit defined over a path or cycle is combinatorially fault-tolerant for any number of preparation errors.
\end{theorem}
\noindent
We give the elementary proof in Appendix~\ref{app:tolerant}.

We have now described all of our constructions.
In the next section we present data on their performance.

\section{Performance of purification circuits}
\label{sec:performance}

\begin{table}[h]
\centering
\footnotesize
\begin{tabular}{p{9em}ccccc}\toprule
\centering name                                    & leading order error probability                             & $n/k$ & gates$/k$ & depth & figure                           \\\midrule
\centering bare                                    & $\prep$                                                              & 1   & 0  & 0  &                                     \\
\centering $(3,1,1)$                               & $3\prep^2+4\prep(\idle/2)+(\idle/2)^2+3(\cnot/4)+4(\toffoli/8)$      & 3   & 3  & 3  & \ref{fig:(3,1,1)}                   \\
\centering $(5,1,1)$                               & $\prep^2+8\prep(\idle/2)+16(\idle/2)^2+3(\cnot/4)+4(\toffoli/8)$     & 5   & 6  & 6  & \ref{fig:(5,1,1)}                   \\
\centering $(7,1,2)$                               & $6\prep^3+\prep(\idle/2)+4(\idle/2)^2+3(\cnot/4)+4(\toffoli/8)$      & 7   & 9  & 6  & \ref{fig:(7,1,2)}                   \\
\centering $(9,1,3)$                               & $27\prep^4+12\prep^2(\idle/2)+(\idle/2)^2+3(\cnot/4)+4(\toffoli/8)$  & 9   & 12 & 6  & \ref{fig:composition}               \\\midrule
\centering $(5,1,2)$a                              & $10\prep^3+5\prep(\idle/2)+9(\idle/2)^2+3(\cnot/4)+21(\toffoli/8)$   & 5   & 9  & 9  & \ref{fig:(5,1,2)a}                  \\
\centering $(5,1,2)$b                              & $10\prep^3+(\idle/2)+3(\cnot/4)+15(\toffoli/8)$                      & 5   & 9  & 9  & \ref{fig:(5,1,2)b}                  \\
\centering $(8,2,2)$                               & $(29\prep^3+8(\idle/2)+6(\cnot/4)+42(\toffoli/8))/2$                 & 4   & 9  & 13 & \ref{fig:(8,2,2)}                   \\\midrule
\centering $(7,4,1)$                               & $(2\prep+18(\idle/2)+14(\cnot/4)+20(\toffoli/8))/4$                  & 7/4 & 3  & 8  & \ref{fig:hamming-scheduled}         \\
\centering cycle (small light cone)                & $8\prep^2+(\idle/2)+3(\cnot/4)+4(\toffoli/8)$                        & 2   & 3  & 4  & \ref{fig:small-light-cone}          \\
\centering cycle (large light cone)                & $8\prep^2+(\idle/2)+3(\cnot/4)+4(\toffoli/8)$                        & 2   & 3  & 4  & \ref{fig:large-light-cone}          \\
\centering fault-tolerant cycle (small light cone) & $6\prep^2+(\idle/2)+3(\cnot/4)+12(\toffoli/8)$                       & 2   & 4  & 6  & \ref{fig:small-tolerant-light-cone} \\
\centering fault-tolerant cycle (large light cone) & $6\prep^2+(\idle/2)+3(\cnot/4)+12(\toffoli/8)$                       & 2   & 4  & 6  & \ref{fig:large-tolerant-light-cone} \\\bottomrule
\end{tabular}
\caption{Leading order error probability on each output qubit.  When $k>1$, this is interpreted as the expected number of errors on the output divided by $k$.}
\label{table:first-order}
\end{table}

For each circuit or light cone discussed in Sections~\ref{sec:purification} and~\ref{sec:graphs} we computed the distribution of the output qubits.
The $2^k$ probabilities in this output distribution are polynomials in $\prep, \idle, \cnot, \toffoli$.
Rather than perform Monte Carlo experiments we calculated enough terms of these polynomials to obtain the desired accuracy across the parameter range of interest; we give more details in Appendix~\ref{app:assessing}.

The circuits we simulate are listed in Table~\ref{table:first-order} along with some basic parameters.
To simulate a circuit with idle errors we need to decide how to schedule the gates into rounds.  
For the cycle-based constructions the gates were scheduled manually as described in Figures~\ref{fig:light-cone} and~\ref{fig:tolerant-light-cone}.
For the other constructions the gates were scheduled greedily, commuting gates past each other when necessary.
The circuit diagrams in this paper were programmatically generated from the scheduled circuits as simulated, and typeset using yquant~\cite{yquant}.
The typeset circuits respect the scheduled gate order, but to avoid visual clutter we have not forced them to respect the scheduling of gates into rounds.
The exact sequence of idle, CNOT and Toffoli gates simulated can be generated from the published source code~\cite{sourcecode}.
For example, the $(3,1,1)$ circuit simulated is
\begin{verbatim}
Circuit {outputBits = [0],
         gatesOf = [CNOT 0 1,Idle 2,CNOT 0 2,Idle 1,Toffoli 1 2 0]}.
\end{verbatim}

For a circuit or light cone we write $\pout(\prep,\idle,\cnot,\toffoli)$ for the expected number of errors on the output, divided by the number of outputs.
When there is a single output this is the probability of an error on that qubit;
when there are multiple outputs this is the average probability of an error on each qubit.

Let $g_1, g_2, g_3$ be the number of idle, CNOT and Toffoli gates in a circuit.
The polynomial $\pout(\prep,\idle,\cnot,\toffoli)$ is a sum of terms of the form 
\[
\prep^{f_0}(1-\prep)^{n-f_0}(\idle/2)^{f_1}(1-\idle)^{g_1-f_1}(\cnot/4)^{f_2}(1-\cnot)^{g_2-f_2}(\toffoli/8)^{f_3}(1-\toffoli)^{g_3-f_3},
\]
each corresponding to a run of the circuit in which there were $f_0$ preparation errors, $f_1$ idle errors, $f_2$ CNOT errors and $f_3$ Toffoli errors.
In particular,
\begin{equation}\label{eqn:perfect}
\pout(\prep,0,0,0) = \sum_{i=1}^n a_i \prep^i(1-\prep)^{n-i}
\end{equation}
for some $a_i \geq 0$, where the $i=0$ term is $0$ as there are no errors on the output if there are no errors on the input.
We plot these polynomials for each circuit in Figure~\ref{fig:perfect}.
Note that the small and large light cones of each cycle construction, and the two $(5,1,2)$ circuits, are equivalent in the absence of idle and gate errors, so only one representative of each type is plotted in Figure~\ref{fig:perfect}.

As expected, near $0$ the performance of the circuits is determined by the leading order term of their error probability.
We include the plot for the full range of preparation errors to point out some general features.
\begin{itemize}
\item 
All curves pass through $(0,0)$, as the zero state is fixed by CNOT and Toffoli gates.
\item
All curves pass through $(0.5,0.5)$, as a uniform mixture of the basis states after preparation is preserved by every gate.
For the majority vote-based $(2e+1,1,e)$ circuits, the curves are symmetric about this point.
\item 
At $\prep=1$ the output is deterministic.
For most circuits, such as the majority vote $(2e+1,1,e)$ circuits, the output is $\ket 1$, an error. 
For the cycle construction, an all $\ket 1$ input happens to produce an all $\ket 0$ output.
(For $\prep < 0.5$ the error rate at $1-\prep$ is greater than the error rate at $\prep$, so this isn't an argument in favour of preparing $\ket 1$ states deliberately.)
The $(7,4,1)$ circuit passes through $(1,0.75)$, meaning that an all $\ket 1$ input produces one $\ket 0$ and three $\ket 1$ states on the output.
\end{itemize}

\begin{figure}
\centering
\includegraphics{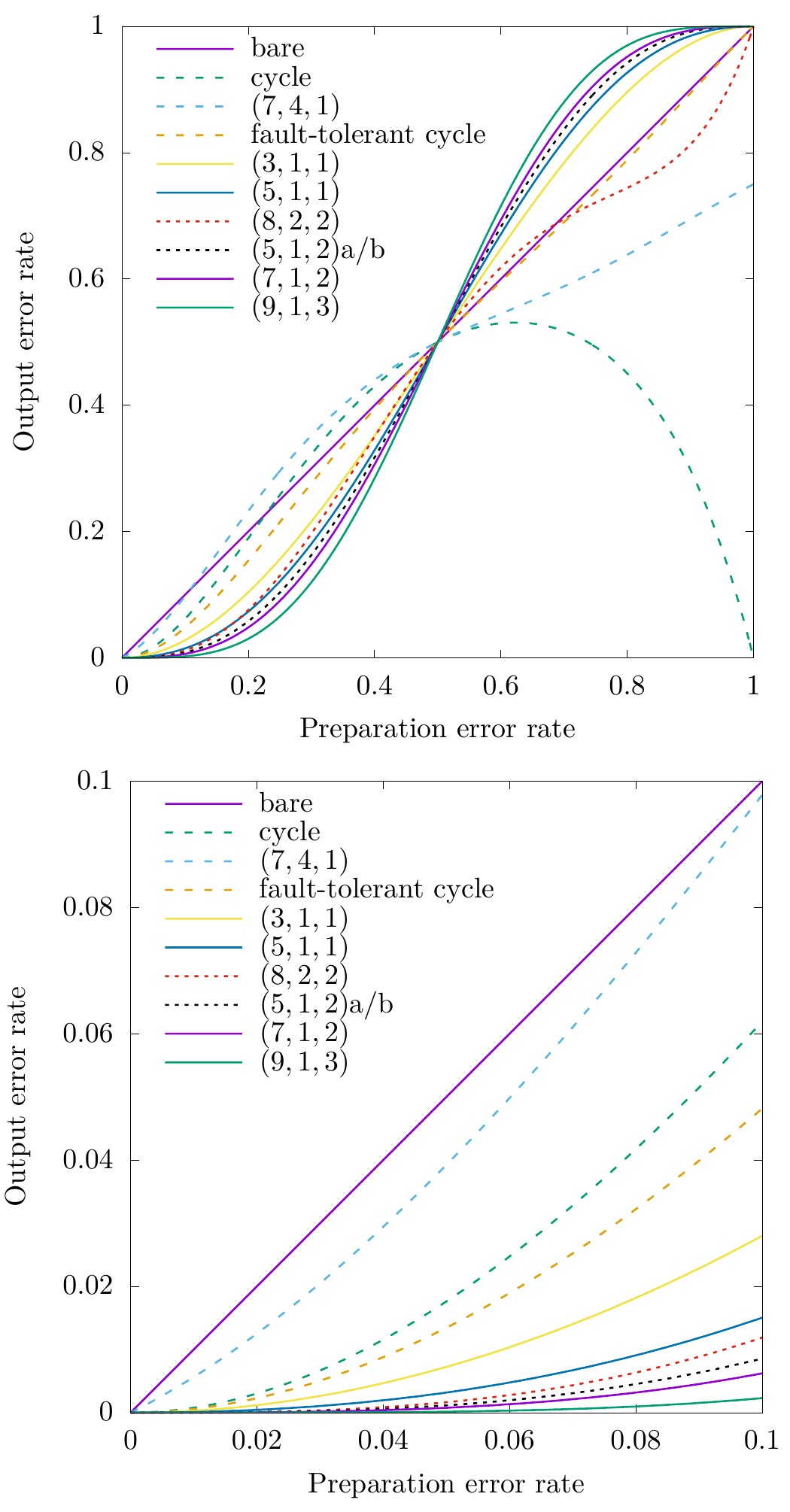}
\caption{Circuit performance with perfect gates.  When $k>1$, output error rate is interpreted as expected number of errors divided by $k$.}
\label{fig:perfect}
\end{figure}

\subsection{Noisy gates}
\label{sec:noisy-gates}

\begin{figure}
\centering
\includegraphics{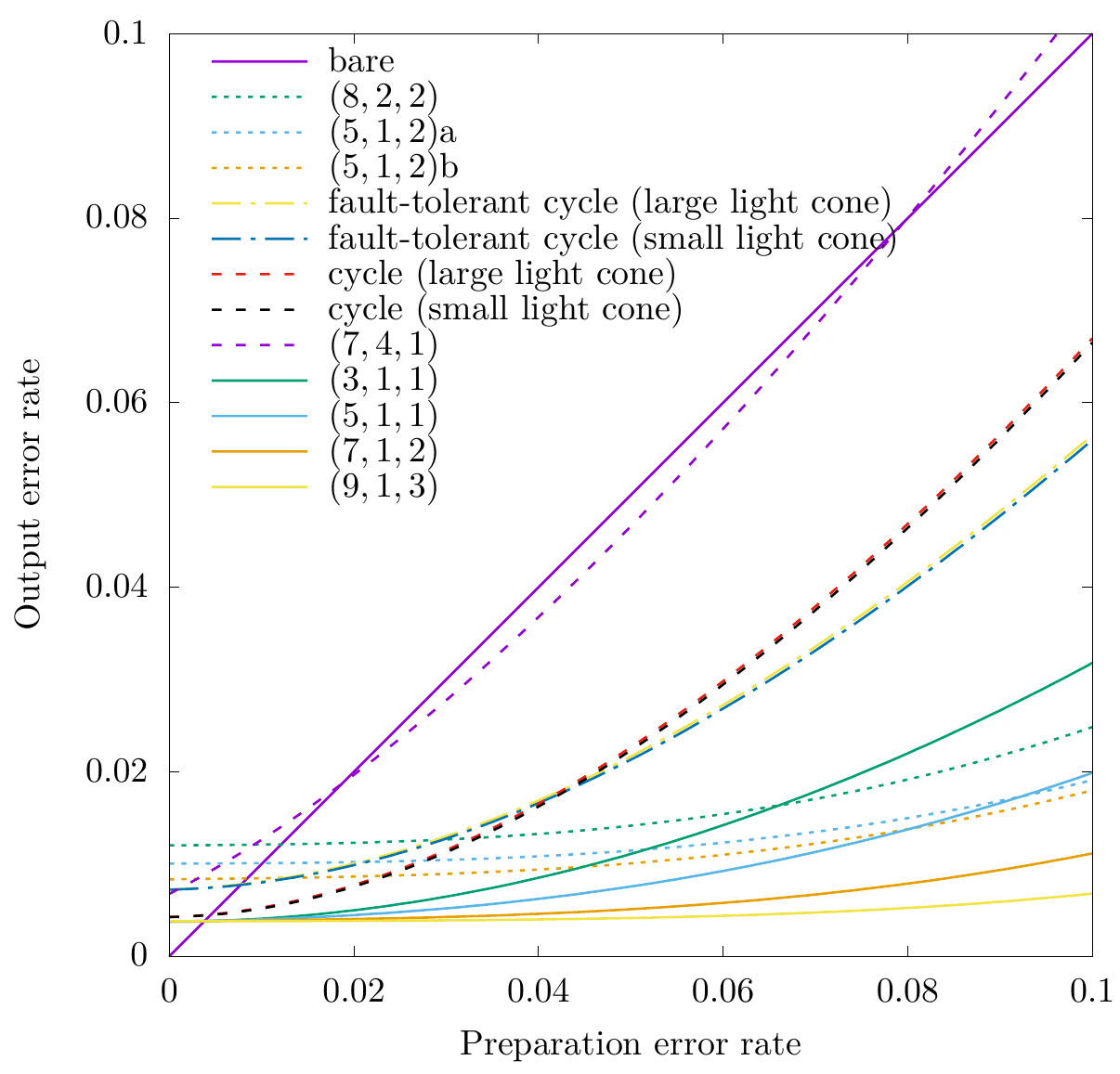}
\caption{Circuit performance with noisy gates ($\idle = 0.001, \cnot=\toffoli=0.003$).  When $k>1$, output error rate is interpreted as expected number of errors divided by $k$.}
\label{fig:noisy}
\end{figure}

Throughout this section we fix an idle error rate $\idle=0.001$.

Figure~\ref{fig:noisy} shows the performance of the various circuits with representative gate noise parameters $\cnot = \toffoli = 0.003$.
The main difference from Figure~\ref{fig:perfect} is that there can now be errors on the output even without errors on the input.
This limits the performance of circuits, moving the $y$-intercepts in the plot away from $0$.
For $\prep$ close to $0$ we have $\pout > \prep$, as high quality qubits are damaged by noisy gates.
Then for each circuit there is an interval of $\prep$ for which the circuit reduces the output error rate below $\prep$.
This interval might extend as far as $\prep = 0.5$ (at which no purification circuit can lead to an improvement) or, as in the case of the $(7,4,1)$ circuit in Figure~\ref{fig:noisy}, end much earlier.

These intervals are sections through the region of parameter space on which using each circuit leads to an improvement in average qubit quality.
Figure~\ref{fig:region} illustrates this region for the $(7,4,1)$ circuit.
Figure~\ref{fig:(3,1,1)region} illustrates this region for the $(3,1,1)$ circuit.
Note that these regions are genuinely different in form.
For example, immediately below the plane $\prep = 0.5$ there is no improvement from running the $(7,4,1)$ circuit, but there is an improvement from running the $(3,1,1)$ circuit.

\begin{figure}
\centering
\subcaptionbox{
The region of parameter space ($0\leq \prep \leq 0.5$, $\idle=0.001$, $0\leq \cnot \leq 0.05$, $1 \leq \toffoli/\cnot \leq 3$) in which using the $(7,4,1)$ circuit leads to an improvement in average qubit quality.
Each data point represents a set of parameters for which $\pout = \prep$.
The plane at $\prep=0.5$ corresponds to the maximum entropy state in which every computational basis state is equally likely.
The curved surface bounds the region to its left in which running the circuit provides an advantage over naive preparation.
\label{fig:(7,4,1)region}}
{
\includegraphics{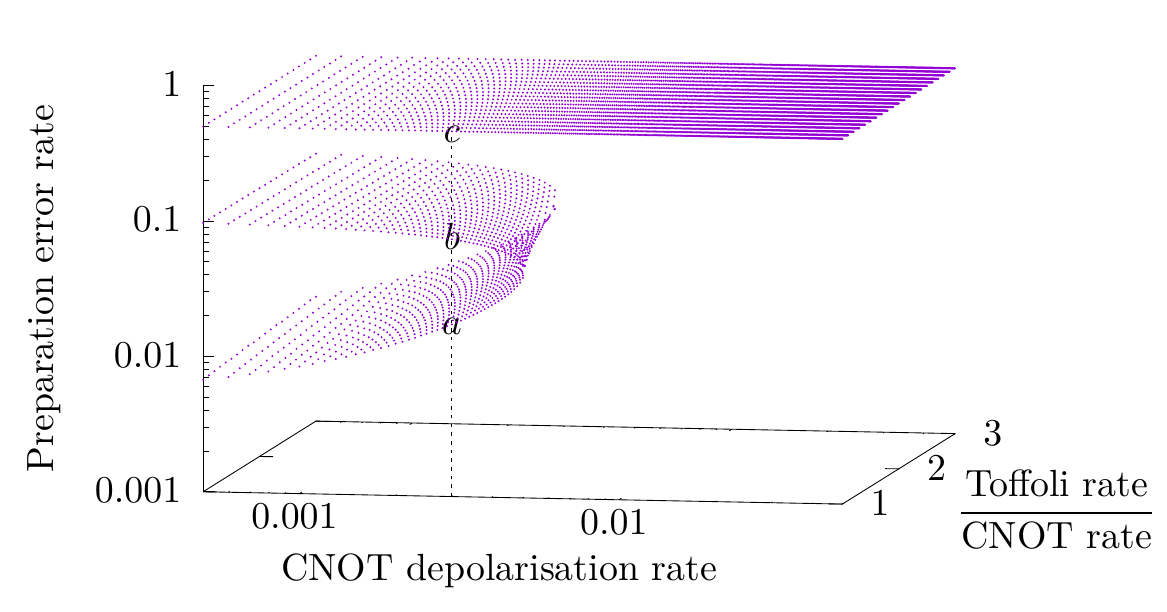}
}
\subcaptionbox{Each vertical line through Figure~\ref{fig:(7,4,1)region} intersects the surface exactly where the corresponding curve $\pout(\prep) = \pout(\prep,0.001,\cnot,\toffoli)$ crosses the line $y=x$.
The marked points $a \approx (0.019,0.019)$, $b\approx(0.080,0.080)$ and $c=(0.5,0.5)$ correspond to the identically marked points in Figure~\ref{fig:(7,4,1)region}.
\label{fig:(7,4,1)line}
}
{
\includegraphics{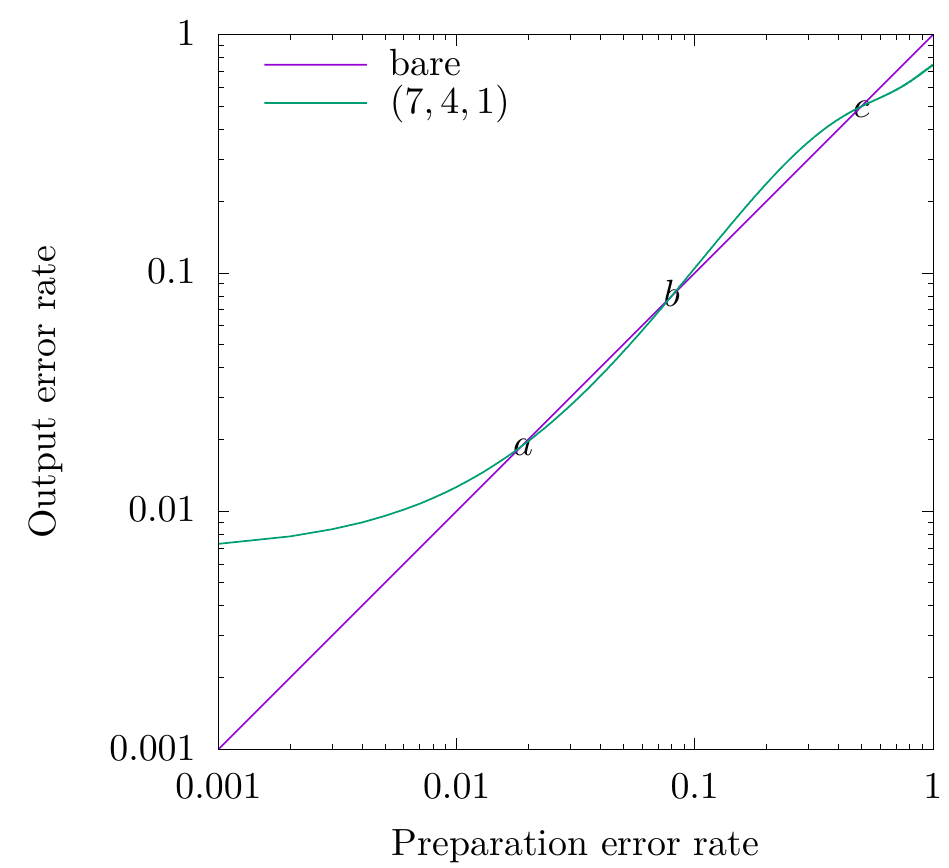}
}
\caption{}\label{fig:region}
\end{figure}

\begin{figure}
\centering
\includegraphics{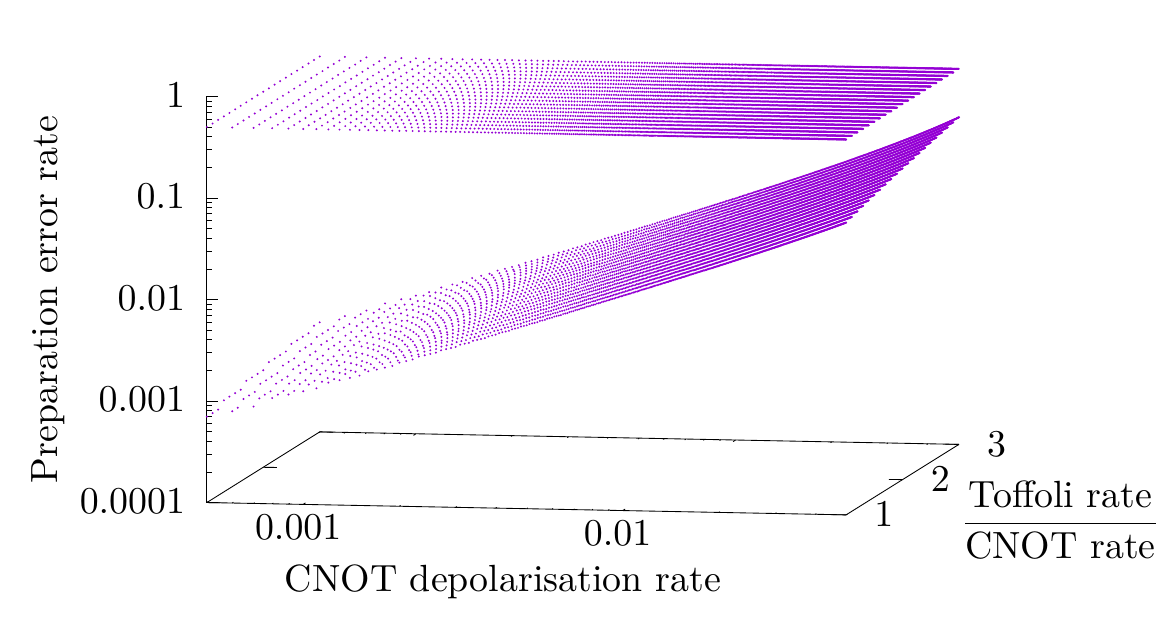}
\caption{Parameter sets between these two surfaces lead to improvements when running the $(3,1,1)$ circuit.}
\label{fig:(3,1,1)region}
\end{figure}

Let $\threshold(\idle, \cnot, \toffoli) = \inf \{\prep : \pout(\prep,\idle,\cnot,\toffoli)=\prep\}$ be the lower threshold preparation error rate for obtaining an improvement.
In the setting of Figures~\ref{fig:(7,4,1)region} and~\ref{fig:(3,1,1)region}, $\theta$ corresponds to the lowest intersection point of each vertical line with the boundary surface.
In particular, $\threshold$ is not necessarily continuous.

Figure~\ref{fig:threshold-3} shows the threshold $\theta(0.001,\cnot,\toffoli)$ for the $(3,1,1)$ circuit.
The $x$-axis shows $\cnot$, the $y$-axis shows the ratio $\toffoli/\cnot$ of Toffoli to CNOT error rate and the $z$-axis, represented by colour, shows $\theta$.
We include contour lines for $\prep \in \{0.003, 0.01, 0.03, 0.1\}$.
If your CNOT and Toffoli error rates place you on a contour, running the $(3,1,1)$ circuit produces output of the same quality as the input at that value of $\prep$.
If your CNOT and Toffoli error rates place you below-left of a contour, you expect to obtain an improvement from running a circuit at that value of $\prep$.
Examination of plots like Figures~\ref{fig:(7,4,1)region} and~\ref{fig:(3,1,1)region} for each circuit show that this holds in practice.
It would be interesting to find general conditions under which output error rate is an increasing function of gate error rate.
We return to this question in Section~\ref{sec:monotonicity}.

Figure~\ref{fig:threshold-3} shows that the $(3,1,1)$ circuit can improve output qubit quality even for high gate failure rates; for example, a CNOT failure rate of 1\% and a Toffoli failure rate of 3\% suffices to improve a preparation error rate of 3\% (although in practice you would like lower gate errors to obtain a meaningful improvement).
The $(5,1,1)$, $(7,1,2)$ and $(9,1,3)$ circuits have the same thresholds as the $(3,1,1)$ circuit: if applying the $(3,1,1)$ circuit reduces error rates, then applying it multiple times also reduces error rates.

\begin{figure}
\centering
\includegraphics[scale=0.95]{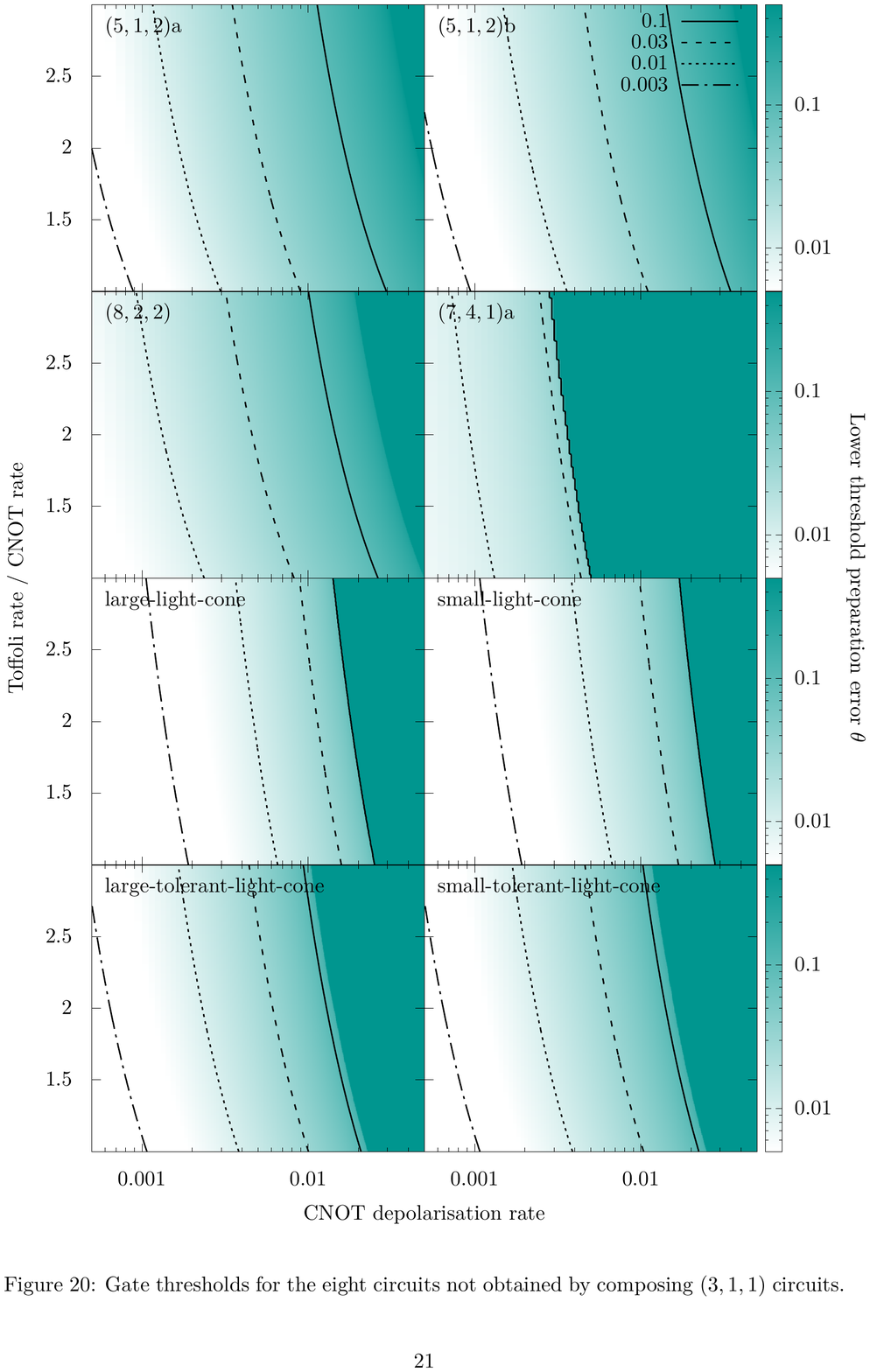}
\caption{
Thresholds for the eight circuits not obtained by composing $(3,1,1)$ circuits.
Each subplot shows, for each value of $\cnot$ and $\toffoli/\cnot$, the least $\prep$ beyond which running the corresponding circuit can produce output of higher quality than the input.
If your gate errors places you down and to the left of a contour, then you expect an improvement from running the corresponding circuit.
See Section~\ref{sec:noisy-gates} for a complete description.
}
\label{fig:thresholds}
\end{figure}

Thresholds for the remaining eight circuits are plotted in Figure~\ref{fig:thresholds}.
Note that many of the thresholds exhibit the discontinuity---a sharp step up to $0.5$---that we should expect from Figure~\ref{fig:(7,4,1)region}.
Figures~\ref{fig:noisy} and~\ref{fig:thresholds} suggest that the basic path/cycle construction of Figure~\ref{fig:graph} provides very good performance at a modest factor of two overhead in number of qubits.

\subsection{Monotonicity}
\label{sec:monotonicity}

We might expect that running a purification circuit with lower gate and preparation error rates always produces a higher quality output, but Figure~\ref{fig:perfect} shows that this is false for preparation errors and the circuit defined over a cycle.

It is, however, true near $0$.
It follows from \eqref{eqn:perfect} that $\pout(\prep,0,0,0)$ is increasing for $\prep < 1/n$.
A similar argument shows that $\pout(0,\idle,0,0)$, $\pout(0,0,\cnot,0)$ and $\pout(0,0,0,\toffoli)$ are increasing when each failure probability is less than one over the corresponding number of gates.
Since $\pout$ is a polynomial, continuity of the partial derivatives shows that it is increasing in some neighbourhood of $0$.

As remarked earlier, when $\big|\binom n 0 + \cdots + \binom n e\big| = 2^{n-k}$ (so that inequality \eqref{eqn:space-bound} is tight) and $\prep < 0.5$, every acceptable input state (pattern of preparation errors that leads to clean output) has higher probability than every unacceptable input state.
This means that any averaging process, for example an arbitrary sequence of depolarisations, is bad for the output error rate.
Again, this is not enough to show that $\pout$ is increasing in each error rate, as successive depolarisations do not necessarily make things progressively worse.

We conjecture that it should be possible to prove monotonicity under some reasonable set of conditions.
For instance, suppose that $\big|\binom n 0 + \cdots + \binom n e\big| = 2^{n-k}$ and $\prep < 0.5$.
Is $\pout(\prep,\idle,\cnot,\toffoli)$ an increasing function of $\prep$, $\idle$, $\cnot$ and $\toffoli$?
We leave this as an open question.

\section{Outlook}
We have described a variety of relatively convenient low depth circuits compatible with 1D arrays of qubits which substantially improve $\ket 0$ preparation when two-qubit errors are well below preparation errors. 
We believe non-Clifford circuits are also worth exploring for pre- and post-processing of the measurement subsystems that will be used in syndrome extraction and elsewhere in quantum computers.

The practical use of these circuits remains an open question, as tradeoffs in qubit cost and connectivity, as well as the cost in time and energy for physical resets of unused qubits, depend substantially upon the physical architecture.
We would very much like to see data on the practical effect of running our circuits on real quantum computing hardware.

\newpage
\appendix

\section{Existence of purification circuits}
\label{app:existence}

In this section we prove Proposition~\ref{prop:existence} on the existence of purification circuits with given parameters.
We use the following algebraic fact.

\begin{lemma}\label{lem:alternating}
Let $G_n$ be the group generated by CNOT and Toffoli gates on $n$ qubits.
For $n \geq 4$, $G_n$ acts as the alternating group on $\mathbb F_2^n\setminus\{0\}$.
For $n \leq 3$, $G_n$ acts as the symmetric   group on $\mathbb F_2^n\setminus\{0\}$.
\end{lemma}

\begin{proof}[Proof of Proposition~\ref{prop:existence}]
By \eqref{eqn:space-bound}, $1 + n \leq 2^{n-k}$, so we may assume that $n > 1$.
Then $3 \leq 2^{n-k}$, and so $k \leq n-2$.

Let $\sigma$ be any permutation of $\mathbb F_2^n \setminus \{0\}$ which maps the vectors of weight at most $e$ into the vectors that vanish in the first $k$ positions.
Let $\tau$ the transposition that swaps the standard basis vectors $e_n$ and $e_{n-1}$.
Since $k \leq n-2$, $\tau\sigma$ also maps the vectors of weight at most $e$ into the vectors that vanish in the first $k$ positions.
One of these permutations is even, so by Lemma~\ref{lem:alternating} can be expressed in terms of CNOTs and Toffolis.
\end{proof}

\begin{proof}[Proof of Lemma~\ref{lem:alternating}]
For $n=1$, all the groups in question are trivial.

For $n \geq 2$, let $C_n$ be the group generated by CNOTs on $n$ qubits.
$C_n$ is isomorphic to $\mathrm{GL}(n,2)$ acting naturally on $\mathbb F_2^n\setminus\{0\}$. 
For $n=2$, this achieves the action of the full symmetric group on $\mathbb F_2^2\setminus\{0\}$. 

For $n \geq 3$, let $T_n$ be the group generated by Toffolis on $n$ qubits.
$T_n$ acts on the vectors of weight at least $2$.
For $n=3$, it acts as the symmetric group $\Sym(4)$; for $n \geq 4$ it acts as the alternating group $\Alt(2^n-n-1)$~\cite[Theorem~1.1.4]{gajewski}.

Now consider the group $G_n = \langle C_n, T_n\rangle$.
\begin{itemize}
\item
The action of $G_n$ on $\mathbb F_2^n\setminus\{0\}$ is $2$-transitive (because $\mathrm{GL}(n,2)$ is $2$-transitive) and therefore primitive.
\item 
The action of $G_n$ on $\mathbb F_2^n\setminus\{0\}$ contains a $3$-cycle ($2$-cycle) if $n\geq 4$ ($n=3$), since $T_n$ does.
\end{itemize}
As $G_n$ acts primitively on $\mathbb F_2^n\setminus\{0\}$ and contains a $3$-cycle ($2$-cycle) if $n\geq 4$ ($n=3$), it contains the alternating group (symmetric group) on $\mathbb F_2^2\setminus\{0\}$~\cite[Theorem 3.3A]{dixmort}. 
A CNOT or Toffoli gate on at least 4 qubits is an even permutation of the basis states, so this containment is equality.
\end{proof}

\section{An explicit construction}
\label{app:explicit}

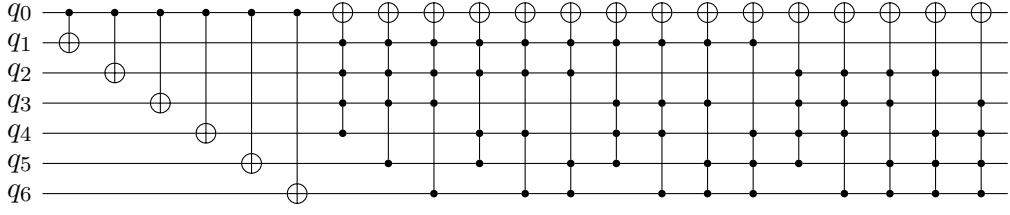
\begin{figure}
\centering
\begin{tikzpicture}
  \begin{yquant}
    qubit {$q_\idx$} q[7];
    cnot q[1] | q[0];
    cnot q[2] | q[0];
    cnot q[3] | q[0];
    cnot q[4] | q[0];
    cnot q[5] | q[0];
    cnot q[6] | q[0];
    cnot q[0] | q[1], q[2], q[3], q[4];
    cnot q[0] | q[1], q[2], q[3], q[5];
    cnot q[0] | q[1], q[2], q[3], q[6];
    cnot q[0] | q[1], q[2], q[4], q[5];
    cnot q[0] | q[1], q[2], q[4], q[6];
    cnot q[0] | q[1], q[2], q[5], q[6];
    cnot q[0] | q[1], q[3], q[4], q[5];
    cnot q[0] | q[1], q[3], q[4], q[6];
    cnot q[0] | q[1], q[3], q[5], q[6];
    cnot q[0] | q[1], q[4], q[5], q[6];
    cnot q[0] | q[2], q[3], q[4], q[5];
    cnot q[0] | q[2], q[3], q[4], q[6];
    cnot q[0] | q[2], q[3], q[5], q[6];
    cnot q[0] | q[2], q[4], q[5], q[6];
    cnot q[0] | q[3], q[4], q[5], q[6];
  \end{yquant}
\end{tikzpicture}
\caption{A $(7,1,3)$ purification circuit presented using multiply controlled Toffoli gates.}
\label{fig:(7,1,3)}
\end{figure}

In this section we give an explicit construction of purification circuits with parameters $(2^{m+1}-1, 1, 2^m-1)$.

Let $n = 2^{m+1}-1$ and $e = 2^m-1$.
Let $q_0$ be an output qubit and $q_1, \ldots, q_{n-1}$ be auxiliary qubits.
Form a circuit in two stages.
\begin{itemize}
\item
In the first stage, perform $n-1$ CNOTs controlled on $q_0$ and targeting each other $q_i$.
\item
In the second stage, perform $\binom{n-1}{e+1}$ Toffolis targeting $q_0$ and controlled on each set of $e+1$ qubits from $\{q_1, \ldots, q_{n-1}\}$.
\end{itemize}

For $m=1$ this is the $(3,1,1)$ circuit in Figure~\ref{fig:(3,1,1)}.
For $m=2$ it is the $(7,1,3)$ circuit shown in Figure~\ref{fig:(7,1,3)};
recall from Section~\ref{sec:small-circuits}  that searching for a shortest $(7,1,3)$ circuit by brute force was out of reach.

As presented, this circuit already has a large number of gates.
If the multi-controlled Toffolis are expanded to conventional Toffolis, this number will increase further, the exact increase depending on how the expansion is performed and the original order of the multi-controlled Toffolis.

\begin{proposition}
This circuit is an $(n,1,e)$ purification circuit.
\end{proposition}

\begin{proof}
Let $0 \leq r \leq e$ and let $v$ be a vector of weight $r$.
We must show that the circuit maps $v$ to a vector which is $0$ in the first position.

If $v_0 = 0$, then every gate of the circuit fixes $v$, so the output $0$ is preserved.

If $v_0 = 1$, then all of the CNOTs activate, mapping $v$ to a vector $w$ with $w_0$ = 1 and $(n-1) - (r-1) = 2e + 1 - r \geq e+1$ of bits $1, \ldots, n-1$ set.
This will activate
\[
\binom{2e + 1 - r}{e+1} = \binom{2e + 1 - r}{e-r} = \binom{2^m + s} s
\] 
of the Toffolis, where $0 \leq s = e-r \leq e < 2^m$.
We claim that this number is odd, from which the result follows.

We use induction on $s$.
For $s=0$ the number is $1$, so odd.
For $s > 0$, we have
\[
\binom{2^m + s} s = \frac {2^m + s} s \binom{2^m + s - 1} {s-1},
\]
where by induction the binomial coefficient on the right-hand side is odd.
Since $s < 2^m$, the number of powers of 2 dividing the numerator and denominator of the fraction agree, so it has the form $a/b$ where both $a$ and $b$ are odd.
Hence the left-hand side is odd, as required.
\end{proof}

\section{Fault-tolerance of the enhanced cycle construction}\label{app:tolerant}

In this section we prove Theorem~\ref{thm:tolerant}, that the enhanced cycle construction is fully fault-tolerant.

\begin{proof}[Proof of Theorem~\ref{thm:tolerant}]
A path on two vertices only has one output, so is automatically fault-tolerant.
Otherwise a path behaves like a cycle on the same number of vertices where one nominated edge never has an error on the input, so we prove the result for cycles.

Let $V$ be the set of vertices with an error on the input, $E$ the set of edges with an error on the input and $F$ the set of edges with an error on the output.
We must show that $|F| \leq |V| + |E|$, or equivalently that
\begin{equation}\label{eqn:injection}
    |F \setminus E| \leq |V| + |E \setminus F|.
\end{equation}

Fix an orientation of the cycle.
We will show that, starting from any element of $F\setminus E$ and walking round the cycle clockwise, we encounter an element of $V \cup (E \setminus F)$ before another element of $F \setminus E$.
This suffices to prove \eqref{eqn:injection}.

Let $xyz$ be a sequence of three consecutive vertices moving clockwise around the cycle.
Suppose that $xy \in  F \setminus E$ but $y \notin V$.
For $xy$ to have an error on the output but not the input, $y$ must be in state $\ket 1$ after the detect$'$ stage.
Since $y$ and $xy$ have no errors on the input, $yz$ must have an error on the input.
If $yz \in E \setminus F$ then we are done, so assume that $yz \in E \cap F$.
But now, by construction, an input error on an edge cannot persist in the output if there are no input errors on its endpoints.
There is no input error on $y$, so we must have $z \in V$.
\end{proof}

\section{Searching for circuits}\label{app:searching}

As described in Section~\ref{sec:existence} and Appendix~\ref{app:existence}, CNOT and Toffoli gates on a set of $n$ qubits act as permutations on the computational basis, which we identify with $\mathbb F_2^n$.
Let $\Pi_n$ be the set of these CNOT and Toffoli permutations.
An $(n,k,e)$ purification circuit is a composition of elements of $\Pi_n$ which maps the Hamming ball $B(0,e)$ of radius $e$ about $0$ into the codimension $k$ subspace $V_{n-k}$ of $\mathbb F_2^n$ comprising the vectors which vanish in the first $k$ positions.
By Proposition~\ref{prop:existence}, such a composition exists provided $|B(0,e)| \leq |V_{n-k}|$.

We can find such compositions as follows.
Let $\powerset(\mathbb F_2^n)$ be the power set of $\mathbb F_2^n$ and, for a set $S$, let $\binom S m$ be the set of subsets of $S$ of size $m$.
For sets of states $S \in \powerset(\mathbb F_2^n)$, sets of sets of states $\mathcal S \subseteq \powerset(\mathbb F_2^n)$ and permutations $\pi \in \Pi_n$, let
\begin{itemize}
\item $\pi(S) = \{\pi(s) : s \in S\}$
\item $\Pi_n (\mathcal S) = \{\pi(S) : \pi \in \Pi_n, S \in \mathcal S\}$.
\end{itemize}

If $\Pi_n^{t_1}(\{B(0,e)\})$ intersects $\Pi_n^{t_2}\big(\binom {V_{n-k}} {|B(0,e)|}\big)$ then, for some $\pi_i, \rho_i \in \Pi_n$ and some $S \subseteq V_{n-k}$ of size $|B(0,e)|$,
\[
\pi_1 \cdots \pi_{t_1} (B(0,e)) = \rho_1 \cdots \rho_{t_2}(S)
\]
whence, using the fact that every element of $\Pi_n$ is self-inverse,
\[
\rho_{t_2} \cdots \rho_{1} \pi_1 \cdots \pi_{t_1} (B(0,e)) = S \subseteq V_{n-k}.
\]
That is, there is an $(n,k,e)$ purification circuit of length $t_1 + t_2$.

We expect to find such an intersection point once
\[
\big|\Pi_n^{t_1}(\{B(0,e)\})\big|
\times\big|\Pi_n^{t_2}\big(\textstyle\binom {V_{n-k}} {|B(0,e)|}\big)\big|
\approx \textstyle\binom {2^n} {|B(0,e)|},
\]
so time and memory requirements scale roughly as the square root of $\binom {2^n} {|B(0,e)|}$.
This contrasts with a naive search over circuits of length $t_1+t_2$, which requires at least an expected $\binom {2^n} {|B(0,e)|}$ time but constant memory.

The first method has an additional advantage over a naive search.
There are typically multiple circuits taking one set of states to another set of states.
This is in part due to the existence of multiple circuits implementing the same permutation (obtained, for example, by commuting gates past each other), and in part due to the fact that distinct permutations of states can have the same action on sets of states (for example, the two inequivalent circuits in Figure~\ref{fig:(5,1,2)} which both act as $(5,1,2)$ purification circuits).
By tracking reachable sets of states rather than circuits, the first method experiences a limited overhead from this phenomenon.
A brute force search by contrast might be able to avoid the most obvious cases of repetition (for example, by not applying a pair of 
commuting gates in both orders), but greater care would be required to avoid duplicated effort.

In~\cite{sourcecode}, \texttt{Search.hs} implements this method by applying generic pathfinding functions from \texttt{Pathfinding.hs} to the gate set described in \texttt{Circuits.hs}.

\section{Assessing circuits}\label{app:assessing}

We want to assess the performance of a purification circuit with $n$ qubits, $g_1$ idle gates, $g_2$ CNOTs and $g_3$ Toffolis.
With the error model described in Section~\ref{sec:error-model}, at every point during the operation of the circuit the quantum state can be described by a probability distribution $\pi$ over the computational basis states, which we identify with $\mathbb F_2^n$.
Immediately after preparing all $n$ qubits this distribution is
\begin{equation}\label{eqn:prob-initialisation}
\mu_0(x) = \prep^{\| x \|} (1-\prep)^{n - \| x \|},
\end{equation}
where $\|x\|$ is the Hamming weight of $x$, the number of coordinates in which it takes the value~$1$.

Suppose that we have applied $m$ gates, achieving a distribution $\mu_m$ over $\mathbb F_2^n$, and are about to apply a gate $G$ acting on the set of qubits $Q$.
Let $p_G$ be the probability that $G$ depolarises the qubits in $Q$.
Then
\begin{equation}\label{eqn:prob-iteration}
\mu_m(x) = (1-p_G) \mu_m(G^{-1}(x)) + \frac{p_G}{2^{|Q|}}\sum_{v \in \mathbb F_2^Q} \mu_m(x+v).
\end{equation}
Performing this iteration computes the distribution over the computational basis at the end of the circuit in $O(2^n(g_1 + g_2 + g_3))$ polynomial arithmetic operations.
The polynomials in question have up to $(g_1+1)(g_2+1)(g_3+1)$ terms, so in practice the additional overhead from polynomial arithmetic can be significant.

To reduce the overhead we make two choices.
From~\eqref{eqn:prob-initialisation} and~\eqref{eqn:prob-iteration} it follows that $\mu_m$ is a sum of terms of the form
\[
\prep^{f_0}(1-\prep)^{n-f_0}(\idle/2)^{f_1}(1-\idle)^{g_1'-f_1}(\cnot/4)^{f_2}(1-\cnot)^{g_2'-f_2}(\toffoli/8)^{f_3}(1-\toffoli)^{g_3'-f_3},
\]
where $g_1'$, $g_2'$ and~$g_3'$ are the number of idle, CNOT and Toffoli gates applied so far (thus $m = g_1'+g_2'+g_3'$).
Using this basis rather than the standard monomial basis means that the first polynomial multiplication in \eqref{eqn:prob-iteration} is implicit, and the second is an increment of $f_1$, $f_2$ or~$f_3$.

The second choice is to observe that large values of $f_i$ correspond to events that are unlikely, depending on the values of $n, g_1, g_2, g_3, \prep, \idle, \cnot, \toffoli$ and our desired accuracy.
So we need only record terms where each $f_i$ is in some restricted range.

In~\cite{sourcecode}, \texttt{Simulate.hs} implements this method over the description of circuits in \texttt{Circuits.hs}.

\bibliographystyle{plainnat}
\bibliography{purification}

\end{document}